\documentclass[graybox,square]{svmono}

\usepackage{type1cm}        
\usepackage[ruled]{algorithm}
\usepackage{nicefrac}
\usepackage{graphicx}        
\usepackage[none]{hyphenat}                           
\usepackage[bottom]{footmisc}

\usepackage{newtxtext}       %
\usepackage[varvw]{newtxmath}       

\usepackage{hyperref}
\usepackage{cprotect}
\usepackage{cleveref}

\DeclareMathOperator{\polyy}{\text{poly}}

\makeindex             

\begin{document}
	
	\guidelinedefn%

	\title*{A Fourier analytic approach to Gaussian mixture learning
		}
	\author{Somnath Chakraborty$^{(a)}$, Hariharan Narayanan$^{(b)}$}
	\institute{School of Technology and Computer Science,\\
	Tata Institute of Fundamental Research,\\
	Mumbai 400005, India\\
	ORCIDs: $(a)$ 0000-0002-9854-0181, $(b)$  0000-0002-4331-8794\\
	Email: $(a)$ somnath.chakraborty@alumni.iu.edu, $(b)$ hariharan.narayanan@tifr.res.in
}

	\maketitle
	
	\abstract{ Suppose that we are given independent, identically distributed random samples $x_1,\cdots,x_n$ from a mixture at most $k$ many $d$-dimensional spherical Gaussian distributions $\mu_1,\cdots,\mu_{k_0}$ of identical and known variance $\sigma^2$ in each coordinate, such that the minimum $\ell^2$ distance between two distinct centers $y_l$ and $y_j$ is greater than $2\Delta\sigma \min\{\sqrt{d},\sqrt k\}$, where $\Delta>C_0$, and $C_0$ is a sufficiently large universal constant.  We develop a randomized algorithm that learns the centers $y_l$'s of the Gaussian components to within an $\ell^2$ distance of $k^{-\tilde C_0}$ --- in presence of arbitrarily large number of components and in arbitrary dimension, when the weights are known to be uniform. Furthermore, if the number of components is $k= \Omega(2^d)$, then for arbitrary universal constant $c>0$, even for unknown weights, the algorithm learns the centers to within an $\ell^2$ distance of $d^{-\tilde C_0}$ and the weights up to an accuracy of $cw_{min}$, with probability greater than $1 - \exp(-k/c)$, provided that the weights lie in $[c/k,1/ck]$, and the minimum separation is just $2c\sqrt d$. The number of samples and the computational time is bounded above by $\polyy(k, d)$ in either case. Such a bound on the sample and computational complexity was previously unknown in the regime of non-constant dimension, and in particular, when $d$ is not $O(1)$, 
When $d = O(1)$, this complexity bound follows from \cite{MR3734220}, where it has also been shown that the sample complexity of learning a random mixture of Gaussians in a ball of radius $o(\sqrt{d})$ in $d$ dimensions, when $d$ is $\Theta( \log k)$, is at least super-polynomial in $k, d$, showing that our result is tight in this case.
	}

	\def\thesection{\arabic{section}}
	\section{Introduction}
Designing efficient algorithms that estimate the 
parameters of an underlying probability distribution is a central theme in statistical learning theory.
An important special instance of this learning problem is the case when the underlying distribution is known 
to be a finite mixture of Gaussian distributions in $d$-dimensional 
Euclidean space. Such mixtures are popular models for 
high-dimensional data clustering, and learning mixture of Gaussians 
in an unsupervised setting has been a topic of intensive research for the 
last few decades.

In its most general form, the underlying problem is as follows: we have access 
to random samples drawn independently from some Gaussian 
mixture $\mu:=\omega_1\mu_1+\cdots+\omega_{k_0}\mu_{k_0}$, 
where $(\omega_1,\cdots,\omega_{k_0})$ is a probability vector with strictly positive components, 
and each $\mu_l$ is a Gaussian density in $\mathbb R^d$, with mean 
$y_l\in\mathbb R^d$ and covariance $\Sigma_l$. The algorithmic task is to 
estimate each component of the parameter set $\{(\omega_1,y_1,\Sigma_1),\cdots,
(\omega_{k_0},y_{k_0},\Sigma_{k_0})\}$ of the density function 
$\mu$, within a pre-specified accuracy $\epsilon>0$, with success probability at 
least $1-\delta$ for a pre-specified $\delta>0$. For the purposes of 
this paper, we will restrict ourselves to the case where all the 
Gaussian components are spherical, with identical variance in each coordinates.

Recently, \cite{MR3826226} and \cite{MR3826314} 
devised polynomial time learning algorithms that work
with minimum separation of $k^\epsilon$ (for any $\epsilon>0$), although, the sucess 
of these algorithms appear to be guaranteed only in some restricted 
region in the $(d,k)$-space (see \cref{table:1}). 
In \cite{MR3734220}, Regev-Vijayaraghavan considered the question of obtaining a  lower bound on the
separation of centers necessary for the polynomial learnability of Gaussian mixtures. They 
devised an \emph{iterative algorithm for amplifying the accuracy of parameter estimation} that, 
given initializer $y_1^0,\cdots,y_{k_0}^0$ and a desired accuracy parameter $\epsilon>0$, uses 
polynomially many samples from the underlying mixture and computation steps to return $y_1',\cdots,
y_{k_0}'$, that lie within Hausdorff distance (see \cref{definition:haus}) at most $\epsilon$ from the 
true centers; for more details, see Theorem 4.1 of 
\cite{MR3734220}. One of their results establishes that, in constant 
dimension $d = O(1)$, with minimum separation at least $\Omega(1)$, any uniform mixture 
of spherical Gaussians can be learned to desired accuracy $\delta$ (with high 
probability) in number of samples and computation time depending polynomially on the 
number of components and $\epsilon^{-1}$.

In the present paper we answer the question of learning the centers of uniform mixture of spherical Gaussians, 
up to accuracy $k^{-\tilde C_0}$ for arbitrary large constant $\tilde C_0>0$ --- where the minimum 
separation of the centers is $2\Delta\sigma\min\{\sqrt d,\sqrt k\}$ --- 
in number of samples and computational time that depends polynomially 
on the ambient dimension $d$, and the number of components $k$, provided that 
$\Delta$ is larger than a certain absolute constant $C_0$. Moreover, when $d=O(\log k)$, 
our algorithm can recover the centers within accuracy $d^{-\tilde C_0}$, even if 
the minimum separation of the centers is $2c\sigma\sqrt d$ for any absolute constant $c\in (0,1)$. 
We note that our minimum separation requirement is independent of $k$, the number of components, 
and for $d=O(\log k)$, this is the weakest 
minimum separation required for polynomial learnability of the mixture: see theorem 3.1 of \cite{MR3734220}. 
We employ deconvolution in the Fourier domain --- a heavily-studied 
technique in the realm of statistical data analysis; see, for example, \cite{MR1790011}, 
or the references therein. However, 
a carefully chosen cutoff has to be used in course of performing this deconvolution, because deconvolution by a Gaussian is an unstable operation which --- when applied to functions that are not in the trajectory of a heat flow --- does not give rise to functions. Here is the main theorem of this paper: \\


{\bf {Main Theorem}}: \emph{Let $\tilde C_0,c>0$ be arbitrary absolute constants --- with 
$c\in (0,1)$, and let $\Delta>C_0$ for a sufficiently large absolute constant $C_0>0$. 
Given a mixture of at most $k$ standard spherical Gaussians in 
$\mathbb R^d$ --- having identical known variance $\sigma^2$ in each direction --- 
for which $a)$ all the mixing weights are in $[ck^{-1}, (ck)^{-1}]$, 
and $b)$ the minimum separation of the centers of the 
components satisfies} \[\min_{1\leq l\neq j\leq k_0}
||y_l-y_j||_2\geq 2\Delta\sigma\min\{\sqrt d,\sqrt k\}\,.\] \emph{There 
is an efficient randomized algorithm that accomplishes the following task: 
\begin{enumerate}
\item if the weights are known to be uniform, then in arbitrary dimension $d$, the algorithm recovers (with 
high probability) the centers of the mixture, up to accuracy 
$k^{-\tilde C_0}$ in time (and samples) $\mbox{poly}\left(k,d\right)$;
\item if the weights are unknown, but the number of components satisfies $k\geq 2^d$, 
then the algorithm recovers (with high probability) the centers and the weights 
up to accuracy $d^{-\tilde C_0}$ and $cw_{min}$, respectively, 
in time (and samples) $\polyy\left(k,d\right)$, even if the minimum 
separation is $2c\sigma\sqrt d$.\qed
\end{enumerate}}

\vspace{0.25cm}

Note that, by scale invariance nature of the problem, we may assume $\sigma=1$ 
without loss of generality. We will denote the actual number 
of Gaussian components in the mixture by $k_0$, so that $k_0\leq k$. 
It is evident that the results in this paper generalize 
the main upper bound in theorem 5.1 of \cite{MR3734220} 
to the regime of dimension $k=\Omega(2^d)$ and arbitrary 
dimension ($d$), when the (identical) variance 
and uniform weights of the underlying mixture is known, as was also the underlying 
assumption in \cite{MR3734220}. 

\subsection{Earlier works on Gaussian mixture learning}
Many of the earlier approaches to the problem of Gaussian mixture learning were based 
on local search heuristics, $e.g.$ the EM algorithm and $k$-means 
heuristics, that resulted in weak performance guarantees. 
There is a very extensive body of work on this question, and it would 
take a large survey to cover all of it. We make a modest attempt at 
surveying some of the results in this area that are most in line with our results; 
for more elaborate history on some interesting lines of works, 
see (for example) \cite{MR2115036}, \cite{MR3385380}, 
and the references therein.

In \cite{MR1917603}, Dasgupta presented {the first provably correct algorithm} for learning 
a mixture of Gaussians, with a common unknown covariance matrix, 
using only polynomially many (in dimension as well as 
number of components) samples and computation 
time, under the assumption that the minimum separation between the 
centers of component Gaussians is at least $\Omega(\polyy\log(kd)\sqrt d)$. In a consequent work, 
\cite{MR2320668} showed a variation of EM to work with minimum 
separation only $\polyy\log(kd)d^{\frac 14}$, when the 
components are all spherical. Subsequently, many more algorithms 
(see \cite{MR3024779} and the references therein) with improved sample complexity 
and/or computational complexity have since been devised, that work on 
somewhat weaker separation assumption; in particular, the SVD-based algorithm of 
\cite{MR2059647} learns a mixture of spherical Gaussians with poly-sized 
samples and polynomially many steps, under the separation 
assumption \begin{displaymath}\min_{1\leq l\neq j\leq k}||y_l-y_j||_2\geq
C\max\{\sigma_l,\sigma_j\}\left(\min(k, d)^{\frac 14}\log^{\frac 14}\left(dk/\epsilon\right)+
\log^{\frac 12}\left(dk/\epsilon\right)\right)\,.\end{displaymath} 
Here $\epsilon>0$ denotes the desired $\ell^2$-accuracy to which the centers are learned. 
We note that when $k=\Theta(2^{2^d})$ and $\epsilon=k^{-C_{200}}$ for some 
constant $C_{200}>0$, the above separation requirement 
translates to minimum separation much larger than $d^{\frac 12}$ --- this is due to 
the presence of $\log^{\frac 12}(dk/\epsilon)$, which dominates the rest in 
the regime $k=\Omega(2^{2^d})$; whereas, when $k=\Theta(2^d)$, 
the separation requirement is of order $\sqrt d$. 
On the other hand, consideration of Gaussian 
concentration phenomenon in higher dimension intuitively suggests 
that minimum separation of order $\sqrt d$ should help in clustering the samples 
into respective components and therefore, should help in tackling the problem 
in arbitrary dimension.

In another line of work (see \cite{MR2743304} and \cite{MR3024779}, for example) 
the question of polynomial learnability of arbitrarily separated mixture of Gaussians 
(and more general families of distributions have been investigated by 
\cite{MR3366914}), it has been established that, for Gaussian mixture with 
a fixed number of components and the components having a known minimum 
statistical separation ($i.e.$, a minimum total variation distance of some constant 
$\eta>0$), there is an algorithm that runs in polynomial time and 
uses polynomially many samples to learn the parameters of the mixture to any desired 
accuracy, with arbitrarily high probability.

The table below contains an explicit comparison of 
the guarantees derived (when the covariance matrices are all diagonal 
and identical) --- and the assumptions under 
which they are proved to be true --- in this paper to the ones in some of 
the most recent works that we are aware of.

\begin{table}[ht]
\caption{table}{Comparisons with previous works}\label{table:1}
\centering
\begin{tabular}{c c c}
\hline\hline
References & separation, regime & complexity gurantees \\ [0.5ex]
\hline
\cite{MR3734220}& $C\sqrt d$, $d=O(1)$ & $\polyy\left(k,d\right)$\\
\cite{MR2059647}& $C\log^{\frac 12}\left(dk/\epsilon\right)$, all $d,k$ & $\polyy\left(k,d,\epsilon^{-1}\right)$\\
\cite{MR3826226}& $C\sqrt{\log(k/\epsilon)}$, all $d,k$ & $\polyy\left(d,k, \frac 1{\epsilon}\right)$\\
\cite{MR3826314}& $k^\epsilon$, $k=poly(d)$ & $\max\{(dk)^{\frac 1{\epsilon^2}}, k^{O(1)}d^{O\left(\frac 1\epsilon\right)}\}$\\
This paper & $C_0\min\{\sqrt d,\sqrt k\}$, all $d, k$ & $\polyy(k,d)$\\
This paper & \hspace{0.5cm} $c\sqrt d$, $c>0$ arbitrary, $k=\Omega(2^d)$ & $\polyy(k,d)$\\
\end{tabular}
\label{table:1}
\end{table}

\subsection{Our contributions} Our main contributions can be summarized as follows:
\begin{enumerate}
\item We establish that, when the component 
Gaussians are spherical --- with identical known covariance --- the minimum separation needed to learn the 
centers in polynomial time, is of order $\min \{\sqrt d,\sqrt k\}$. 
\begin{enumerate}
\item In particular, when 
the number of component Gaussians in the mixture overwhelms the ambient 
dimension ($k=\Omega(2^d)$), the minimum separation needed to learn the centers, within 
inverse polynomial (in $k$) accuracy, is $O(\sqrt d)$, which  
improves upon all the previously known separation requirements with respect to the constant overhead, 
as indicated in \cref{table:1}.
\item Moreover, in the setting of $d=O(\log k)$, this 
is the best possible, as established by (theorem III.2 of) \cite{MR3734220}.
\end{enumerate}
\item We have made a novel application of Fourier analysis --- to the problem of 
learning mixture of Gaussians; to the best of our knowledge, 
this has not been rigorously considered before. An earlier attempt on deconvolution via 
Fourier analysis appeared in \cite{MR1790011}, which addressed a related but different 
problem, based on different techniques. The success of our techniques relies 
on the truncation in Fourier domain (see \cref{subsection:001}, especially 
\cref{eq:supp0} and \cref{lem:3}), which is the crux of the methodological contributions 
in this paper.
\end{enumerate}

\subsection{Organization of the paper}
In \cref{overview}, we give a brief overview of the key 
technical ideas in the paper. Next, \cref{prepro} contains preprocessing 
and reduction steps, whereby we reduce the problem to one where the centers 
of the mixtures are all contained in a ball of radius $\Omega(k\sqrt {d\ln (C_1dn)})$, 
in a $d=O(\log k)$ dimensional space, for some $n=\polyy(d,k)$. 
Then, \cref{fourierdeconvolve} contains bulk of the technical results needed 
to carry out the Fourier deconvolution, which results in obtaining 
black-box oracle access to good additive approximation of the underlying 
atomic mixture. Next, \cref{www} contains 
weight estimation. The final \cref{conclusion} 
briefly mentions a problem that still remains open (to the best of our knowledge).\\

\section{Outline of techniques used}\label{overview}
We observe that if two clusters of centers are very far ($i.e.$, the minimum distance 
of a center in one cluster is at least $ \sqrt{d}$ away from every center of the 
other cluster), then the samples are unambiguously of one cluster only. This is 
shown in \cref{lem:3.1}, allowing us to reduce the question to the case 
when a certain proximity graph --- defined on the centers --- is connected. 
This lemma, along with theorem C.1 of \cite{MR3734220}, also allows us 
to reduce the problem to dimension $d\leq k$.

In \cref{alg:learnmixture}, we use the standard Johnson-Lindenstrauss lemma and 
project the data in the ambient space to at most $d=O(\log k)$ many carefully chosen subspaces of 
dimension at most $O(1 + \min(\log k, d))$, and show that if we can separate 
the Gaussians in these subspaces, the resulting centers can be used to obtain 
a good estimate of the centers in the original question. Thus the question is 
further reduced to one where the dimension $d = O(\log k)$.

Next, in \cref{lem:4.2}, we show that if the number of samples is chosen 
to be an appropriately large polynomial in $k$, then all the 
$k_0$ centers are with high probability contained in a 
union of balls $B_i$ of radius $2\sqrt{d}$ centered around the $n$ data points. 
This allows us to focus our attention to $n$ balls of radius $2\sqrt{d}$.

The main idea is that in low dimensions, it is possible to efficiently implement 
a deconvolution on a mixture of identical standard Gaussians having variance 
$1$ in each direction, and recover a good approximation to a mixture of Gaussians 
with the same centers and weights, but in which each component 
now has standard deviation $\bar \Delta < c\Delta$, where $2\Delta \sqrt{d}$ 
is the minimum separation between two centers. Once this density is available 
within small $L^\infty$ error, the local maxima can be approximately obtained 
using a robust randomized zeroth order concave optimization method developed 
in \cite{Belloni:2015uc} started from all elements of a sufficiently rich $\ell^\infty$ 
net on $\bigcup_j B_j$ (which has polynomial size by \cref{lem:4.2.8}), 
and the resulting centers are good approximations (i.e.\, within $k^{-{\tilde C_0}/2}$ 
of the true centers in Hausdorff distance) of the projections of true centers with high probability.
We then feed these $d^{-\frac{5}{2}}$ approximate projected centers into the iterative 
procedure developed in \cite{MR3734220} and by Theorem 
4.1 in there, a seed of this quality, suffices to produce in 
$\mathrm{poly}(k, d, \delta^{-1})$ time, a set of centers whose projections are within $k^{-\tilde C_0}$ of 
the true projections, in Hausdorff distance.

The deconvolution is achieved by convolving the empirical measure $\mu_e$ obtained 
from independent random samples from the mixture, with the 
Fourier transform of a certain carefully chosen $\hat \zeta$. The function $\hat 
\zeta$ is, upto scalar multiplication, the reciprocal of a  Gaussian with standard 
deviation $\sqrt{1 - \bar{\Delta}^2}$ multiplied by the indicator of a ball of radius $(\sqrt{\log k}
+ \sqrt{d})\bar\Delta^{-1}$. It follows from \cref{lem:3}, that the effect 
of this truncation ($i.e.$, multiplication by the indicator) on the deconvolution process can be controlled. 
The pointwise evaluation of the convolution is done using the Monte Carlo method.
The truncation plays an important role; without it, the reciprocal 
of a Gaussian would not have well-defined Fourier transform.\\ 

\section{Preprocessing \& Reduction}\label{prepro}
Suppose we are given independent, identically distributed samples 
$x_1,\dots,x_n$ from a mixture $\mu$ of no more than $k$ of $d$-dimensional spherical 
Gaussian distributions $\mu_j$ with variance $1$, such that the minimum 
$\ell^2$ distance between two distinct centers $y_l$ and $y_j$ is greater 
than $\Delta\min\{\sqrt{d},\sqrt k\}$ for some $\Delta\geq C_0$, where $C_0>0$ is a 
sufficiently large universal constant. We write \begin{displaymath}
\mu(x):=(2\pi)^{- \frac{d}{2}} \sum_{l=1}^{k_0}
w_l\exp\left(-\frac 12{||x-y_l||^2}\right)\,,\end{displaymath} 
where 
$k_0\leq k$, and $(w_1,\cdots,w_{k_0})$ is an unknown probability vector such that 
$w_{min}:=\min_{l\in[k_0]}w_l$ satisfies $w_{min}\geq \frac{c}{k}$, and also 
$Y:= \{y_1, \dots, y_{k_0}\}$ is the set of 
centers of the component Gaussians in the mixture 
$\mu$. 
Let us fix $1 - \eta := \frac{9}{10}$, to be the success 
probability we will require. This can be made $1 - \eta_0$ 
that is arbitrarily close to $1$ by the following  simple clustering 
technique in the metric space associated with Hausdorff distance, 
applied to the outputs of $100(\log \eta_0^{-1})$ runs of the algorithm.

\begin{algorithm}

find the median of all the number of centers output by the 
different runs of \cref{alg:learnmixture}, and set that to be $k_0$\;

pick a set of centers $Y$ (that is the output of one of the runs) 
having the property that $|Y| = k_0$ and at least half of the 
runs output a set of centers that is within a Hausdorff distance of 
less than $\frac{\Delta\sqrt{d}}{k^{C}}$ to $Y$. It is readily seen 
that --- provided each run succeeds with probability $(1 - \eta)$ --- this 
clustering technique succeeds in producing an acceptable set of 
centers $\{\hat{\boldsymbol{y}}_1,\cdots,\hat{\boldsymbol{y}}_{k_0}\}$ with probability at least $1 - \eta_0$\;

once the centers are fixed, to determine the weights, we take the median 
of the weights assigned to the nearest  center over all  $100(\log \eta_0^{-1})$ runs 
of the algorithm, where --- in determining the nearest center --- ties are broken arbitrarily\;

{\bf {return}} {$\{\hat{\boldsymbol{y}}_1,\cdots,\hat{\boldsymbol{y}}_k\}$.}\
\caption{Boost}\label{alg:boost}
\end{algorithm}





We record the observation that, {\em if $d>k$, then 
by Principal Component Analysis (PCA), it is possible to find a linear 
subspace $S_k\subseteq \mathbb R^d$ of dimension $k$, such that all the centers $y_l$ are within 
$\frac{\sqrt{d}\Delta}{k^C}$ of $S_k$ in $\ell^2$-norm, with probability at least 
$1-\frac{\eta}{100}$, using $\polyy(d, k)$ samples and computational 
time} (see \textit{Appendix C} of \cite{MR3734220}); moreover, this does not impose 
any constraint on the minimum separation requirement. PCA has been used previously 
in this context (see \cite{MR2059647}).

Now, let  $\{(x_1(y_{l_1}),y_{l_1}), \dots, (x_n(y_{l_n}),y_{l_n})\}$ be a set of $n$ independent 
identically distributed random samples from the mixture $\mu$, generated by first 
sampling --- with probability $w_l$ --- the mixture component having mean $y_l$, and 
then picking $x_l(y_l)$ from the corresponding Gaussian. With probability 
$1$, all the $x_l$'s are distinct, and this is part of the hypothesis in the lemma below, 
a proof of which appears in \cref{p3.1}.

\begin{lemma}\label{lem:3.1}
Let ${\mathcal G}$ be a graph whose vertex set is $X = \{x_1, \dots, 
x_n\}$, in which two vertices $x_l$ and $x_j$ are connected 
by an edge if the $\ell^2$ distance between $x_l$ and $x_j$ is 
less than $2\sqrt{3d\ln (C_1dn)}$. Decompose ${\mathcal G}$ into the 
connected components ${\mathcal G}_1, \dots, {\mathcal G}_r$ of ${\mathcal G}$. Then, the 
probability that there exist $l \neq j$ and $x \in {\mathcal G}_l$, $x' \in 
{\mathcal G}_j$ such that $x,x'$ are both from the same Gaussian components is less than $\eta/100$.\qed
\end{lemma}

Thanks to this lemma, we can now concentrate on any one particular connected component 
of the graph $\mathcal G$; in other words, we may assume that our mixture is 
well-separated, and at the same time, none of the centers is a so-called outlier which 
translates to the case where all the centers are (after application of a deterministic linear shift, 
if required) within an origin-centric ball of radius $2k\sqrt{3d\ln (C_1dn)}$.

The algorithm below helps us monitor the effect of dimension reduction, and the patching-up 
of the learnt projected centers. In what follows, we set $\epsilon$ to be  $k^{-2{\tilde C_0}}$. 
We note that this algorithm is called-in only when $\omega(\log k)\leq d$; 
thus, we are tacitly assuming here that the weights are uniform.

\begin{algorithm}

{\bf{Input:}} IID samples $x_1,\cdots,x_N$ from the mixture of Gaussians.

{\bf{Output:}} Candidate centers $\hat{\boldsymbol{y}}_1,\cdots,\hat{\boldsymbol{y}}_k$ 

\begin{itemize}
\item let $e_1, \dots, e_{d}$ be a uniformly random orthonormal set of vectors\;

\item define $\bar d= \min(d, O(\log k))$ dimensional subspace $A_{\bar d}$ to be the span of $e_1, \dots, e_{\bar d}$; by \cref{lem:j-l}, with probability greater than $ 1- \frac{\eta}{100}$, the distance between the projections of any two centers is at least $\left(\frac{\Delta}{2}\right)\sqrt{\bar d}$\;

\item use the low dimensional Gaussian learning primitive \cref{alg:findspikes} from \cref{ssec:4.1} on the samples $\{\Pi_{\bar d} x_j\}$ to solv the $\bar d + O(1)$ dimensional problem with high probability, if the distance between any two centers is at least $\left(\frac{\Delta}{2}\right)\sqrt{\bar d}$; let $(y_1^{(\bar d)}, \dots, y_{\bar d}^{(\bar d)})$ be the output of \cref{alg:findspikes}\;

\begin{itemize}
\item {\bf {if}} {this fails to produce such a set of centers,}{
go back to 1\;
}{

\item {\bf{else}} pass the output of \cref{alg:boost} obtained by processing $count_{max}$ copies of $M$, to the iterative algorithm of \cite{MR3734220}, which will correctly output the centers to accuracy $\epsilon/k$ with the required probability $1 - \exp(-k/c)$\;
}
\end{itemize}

\item for any fixed $l \geq \bar d + 1$, let $A(l)$ denote the span of $e_1, \dots, e_{\bar d}$ augmented with the vector $e_l$;  suppose that we have succeeded in identifying the projections of the centers on to $A(l)$ for $\bar d + 1 \leq l \leq d$ to  within $ {k}^{-C_3}$ in $\ell^2$ distance with high probability\;

\item together with the knowledge of $y_1^{(\bar d)}, \dots, y_{\bar d}^{(\bar d)}$, and the initial guarantee that (the projected) centers have large mutual distance, this allows us to patch up these projections and give us the centers $\hat{\boldsymbol{y}}_1,\cdots,\hat{\boldsymbol{y}}_k$ to within a Hausdorff distance of ${\delta}$ with high probability.\
\end{itemize}
\caption{LearnMixture}\label{alg:learnmixture}
\end{algorithm}

By the preceding algorithm, it is clear that it suffices to consider the case when $d \leq C_{1.5}\left(\log k \right)$ for an appropriate constant $C_{1.5}>0$ (inherited from application of \cref{lem:j-l} in \cref{alg:learnmixture}).\\

\section{Analysis in the regime $d \leq C_{1.5} \left(\log k \right)$}\label{fourierdeconvolve}
The following lemma help us restrict our search for the centers inside a 
union of ``not-too-many" balls of radius $\sqrt {d}$. Proof appears 
in \cref{p4.2}

\begin{lemma}\label{lem:4.2}
The following statement holds with probability at least 
$1 - \frac{\eta}{100}$: if \begin{displaymath}n\geq 1+\frac kc\log\left\lceil
\frac{300k}{\eta}\right\rceil\left\lceil\log\left\lceil\frac{300k
\log\left\lceil\frac{300k}{\eta}\right\rceil}\eta\right\rceil\right\rceil\,,\end{displaymath} 
and $x_1,\cdots,x_n$ are independent random $\mu$-samples, 
then \begin{displaymath}
\{y_1, \dots, y_k\} \subseteq \bigcup_{l \in [n]} B_2(x_l, 2\sqrt{d})\,.\end{displaymath}\qed
\end{lemma}

We recall the following definition:

\begin{definition}\label{definition:haus}
Given two nonempty subsets $S,T\subset \mathbb R^d$, their Hausdorff distance $d_{\mathcal H}
(S,T)$ is $$\max\{\max_{s\in S}\min_{t\in T}||s-t||_2, \max_{t\in T}\min_{s\in S}||s-t||_2\}\,.$$\qed
\end{definition}

Let  ${\mathcal B}(S,T)$ be the set of bijections between $S$ and $T$. It is evident from the definition above that the following holds when $|S|=|T|$: $$d_{\mathcal H}(S,T)=\inf_{\pi\in {\mathcal B}(S,T)}||\pi(s)-s||_2\,.$$

Let ${\mathcal F}$ denote the Fourier transform. For any function $f \in L^1(\mathbb R^d)$, we write $\hat{ f} = {\mathcal F}(f)$, where \begin{eqnarray}\hat{f}(\xi) = \left(\frac{1}{\sqrt{2\pi}}\right)^d \int_{{\mathbb R}^d} f(x) e^{- \imath \xi \cdot x} dx\,.\end{eqnarray} By the Fourier inversion formula, 
\begin{eqnarray} f(x) = \left(\frac{1}{\sqrt{2\pi}}\right)^d \int_{{\mathbb R}^d} \hat f(\xi ) e^{\imath \xi \cdot x} d\xi\,,\end{eqnarray}

The following properties of Fourier transform operator are standard and easily verified, and hence, 
we ommit their proofs altogether.
\begin{lemma}
$(a)$ Let $\gamma(z):= \left(2\pi\right)^{-\frac d2} e^{-\frac{||z||^2}{2}}$ be the standard Gaussian density in $\mathbb R^d$; then, $\hat \gamma: = {\mathcal F}(\gamma) = \gamma$.

$(b)$ If $g,h\in L^2(\mathbb R^d)$ are such that $f\star g\in L^2(\mathbb R^d)$, then ${\mathcal F}(f\star g)=(2\pi)^{\frac d2}\hat f\hat g$. Here $\star$ denotes the convolution.\qed
\end{lemma}

\subsection{Deconvolution}\label{subsection:001}
We make the following conventions for the rest of this paper.\\

{\bf{Convention:}} Symbol $x.y$ represents decimal number. 
All constants $C_{x.y}$ are absolute constants in $[1,\infty)$, and all 
constants $c_{x,y}$ satisfy $c_{x.y}=C_{x.y}^{-1}$. We allow $C_{x.y}$ to depend on 
$C_{x'.y'}$ if $x.y>x'.y'$.\qed\\

Let $\bar \Delta$ equal  $\Delta/C_{3.2}$. Let $\gamma_\sigma(z) := \sigma^{-d} \gamma(\sigma^{-1}z)$ denote the spherical Gaussian whose one dimensional marginals have variance $\sigma^2$, and note that $\hat\gamma_{\bar\Delta}(z)=\gamma(z\bar\Delta)$. Define \begin{align}\label{eq:supp0}{\mathcal B}:=\left(\sqrt{{C_{3.5}\ln k}}+ \sqrt{d}\right)B_2(0, \bar \Delta^{-1})\,,\end{align} where $B_2(0, \bar \Delta^{-1})$ is the Euclidean ball of radius $\bar \Delta^{-1}$; set $\hat{s}(z) := \gamma\left({ z\bar \Delta }\right){\mathbb I}_{\mathcal B}(z),$ where ${\mathbb I}_S$ the indicator function of $S$. In the lemma that follows, we write $s = {\mathcal F}^{-1}(\hat s)$; proof of the lemma appears in \cref{p3}.

\begin{lemma}\label{lem:3} For all $z\in {\mathbb R}^d$, we have 
$|s(z) - \gamma_{\bar \Delta}(z)| \leq (2\pi{\bar\Delta}^2)^{-\frac d2}k^{- \frac{C_{3.5}}2}$.\qed
\end{lemma}

\subsection{An observation}
Let $x_1, \dots, x_n$ be iid samples from the Gaussian mixture $\mu$. 
Let $\mu_e$ denote the uniform probability measure supported on $\{x_1, \dots, x_n\}$. 
Let $\star$ denote convolution in ${\mathbb R}^d$.  Note that the Fourier convolution 
identity is 
$(2\pi)^{-\frac d2}{\mathcal F}(f\star g)=\hat f\hat g$. 
Let $\hat \zeta := \hat{\gamma}^{-1} 
\cdot \hat{s}$, and $\zeta = {\mathcal F}^{-1}(\hat \zeta)$. We will recover the centers 
and weights of the Gaussians from $\zeta \star \mu_e = (2\pi)^{\frac d2}{\mathcal F}^{-1}(\hat \zeta 
\cdot \hat{\mu}_e)$. The heuristics are as follows.

Let $\nu$ denote the unique probability measure satisfying $\gamma \star 
\nu = \mu$. Thus, 
$\nu = \sum_{j=1}^{k_0} w_j \delta_{y_j}$, 
where $\delta_{y_j}$ is a dirac delta supported on $y_j$. It follows, roughly speaking,
that inside $\mbox{supp}(\hat s)$, we get \begin{eqnarray*}{\hat\nu}(w)&
=&(2\pi)^{-d}{\hat\gamma}(w)^{-1}\mathbb E_{X\approx\mu}\left[e^{-iX\cdot w}\right] 
\ \approx (2\pi)^{-\frac d2}{\hat\gamma}(w)^{-1}\hat\mu_e(w)\\ \Rightarrow\hspace{1cm}
\hat s(w){\hat\nu}(w)&\approx& (2\pi)^{-\frac d2}\hat\zeta(w)\hat\mu_e(w)\end{eqnarray*} 
pointwise, and this should (roughly) yield 
$s\star \nu\approx (2\pi)^{-\frac d2}\zeta\star \mu_e$. 
On the other hand, notice that \cref{lem:3} shows 
that $s\star\nu\approx \gamma_{\bar\Delta}\star\nu$, and because the spikes of $\nu$ are 
approximately (up to scaling) the spikes of $\gamma_{\bar\Delta}\star\nu$, we restrict to learning the spikes 
of $\gamma_{\bar\Delta}\star\nu$ by accessing an approximation 
via $(2\pi)^{-\frac d2}\zeta\star \mu_e$. For notational convenience, 
we write $\xi_e:=(2\pi)^{-\frac d2}\zeta\star \mu_e$.

\subsection{Monte-Carlo and access to deconvolved mixture}

In order to formalize the heuristics above, we proceed as follows. Recall that ${\mathcal B}=
\mbox{supp}(\hat s)$, so that $z\in {\mathcal B}$ if and only if $z\in
\mathbb R^{d}$ satisfies 
$||{\bar\Delta}z||\leq \sqrt{C_{3.5}\ln k}+
\sqrt{d}$.

The following proposition allows us to construct a (random) {\bf black box} oracle 
that outputs a good additive approximation of $\gamma_{\bar \Delta}\star\nu$ at 
any given point $x$. See \cref{p0} for a proof of the proposition.

\begin{proposition}\label{randomsample}
Let $z_1,\cdots,z_m$ be independent, random variables drawn from the 
uniform (normalized Lebesgue) probability measure on ${\mathcal B}$. 
Let $x_1,\cdots,x_n$ be independent $\mu$-distributed random samples. If 
$m=C_{3.6}k^{C_{3.6}}\ln k$, and \begin{eqnarray}\label{sample11}n\geq C_{3.6}k^{C_{3.6}+C_{3.5}}
\ln\left(C_{3.6}k^{C_{3.6}+C_{3.5}}\ln k\right),\end{eqnarray} where $C_{3.6}\geq C_{0.1}
C_{3.5}+C_{1.5}\left(\ln(2\pi)+2\ln\left(\frac{2C_{3.2}}c\right)\right)+C_{3.5}
\left(8+\frac{2C_{3.2}^2}{c^2}\right)$, then for any $x\in\mathbb R^d$, the random variable \begin{eqnarray}\label{blackbox}f_x:=\frac {\mbox{vol}({\mathcal B})}m
\sum_{{l\in[m]}}e^{ix\cdot z_l}\gamma
(z_l{\bar\Delta})e^{\frac{||z_l||^2}2}{\hat\mu}_e(z_l)\end{eqnarray} satisfies the 
inequality 
$\left|(\gamma_{\bar\Delta}\star\nu)(x)- \Re f_x\right|\leq 
3k^{-C_{3.5}}$ 
with probability at least $1-8k^{-C_{3.5}}$: 
here $\Re f_x$ denotes the real part of $f_x$.\qed
\end{proposition}

\subsection{Low-dimensional learning algorithm}\label{ssec:4.1}
In \cref{alg:findspikes} below --- whose analysis will be 
deferred to \cref{an:random}, at any point $z \in  \bigcup_{i}  B_2(x_i, 2\sqrt{d})$, 
we will assume access to the random variable $f_z$ in $\mathbb C$, such that 
\begin{displaymath}{\mathbb P}\left[|f_z - (\gamma_{\bar\Delta}\star\nu)(z)| 
<  k^{- C_4}\right] > 1 - k^{-C_4}\end{displaymath} 
As established by \cref{randomsample}, these $f_z$ can be 
constructed using polynomially many samples and computation steps, in such a 
way that for any subset $\{z_1, \dots, z_m\} \subseteq \mathbb R^d,$ the 
set $\{f_{z_1}, \dots, f_{z_m}\}$ consists of independent random variables. 
Let $diam(Q_\ell)$ denote the $\ell^2$ diameter of $Q_\ell$.

In our implementation, we will employ \emph{the efficient zeroth order stochastic 
concave maximization algorithm}, devised in \cite{Belloni:2015uc}. We denote 
this algorithm of \cite{Belloni:2015uc} as ${\mathfrak A}_0$. In $d$-dimensional Euclidean space, the algorithm 
returns an $\epsilon$-maxima of certain $d^{-1}\epsilon$-approximate $t$-Lipschitz 
concave function, and the number of function evaluations used depends polynomially 
on $d, \epsilon$, and $\log t$. The performance guarantee of algorithm ${\mathfrak A}_0$ is summarized 
in \cref{fact:1} stated below.

\begin{definition}
Let $\mathcal B\subseteq \mathbb R^{\bar d}$ be a convex set, and $\xi>0$. 
A continuous function $\phi:{\mathcal B}\rightarrow\mathbb R_+$ is said to be $\xi$-approximately 
log-concave if there exists a continuous function $\psi:{\mathcal B}
\rightarrow\mathbb R_+$ such that $\log\psi$ is concave, and 
$||\log\phi-\log\psi||_{L^\infty({\mathcal B})}\leq \xi$. We say $\phi$ is approximately 
log-concave if it is $\xi$-approximately log-concave for some $\xi>0$.
\end{definition}

\begin{lemma}[Belloni et al]\label{BLNR}\label{fact:1}
Suppose that ${\mathcal B}\subseteq \mathbb R^{\bar d}$ is a convex subset, and $\chi,\psi:{\mathcal B}
\rightarrow \mathbb R$ are functions satisfying $\sup_{z\in {\mathcal B}}|\chi(z)-\psi(z)|\leq \frac{\epsilon}{\bar d}$. Suppose 
that $\psi$ is concave and $t$-Lipschitz; then algorithm ${\mathfrak A}_0$ returns a point $q\in {\mathcal B}$ 
satisfying $\chi(q)\geq \max_{z\in{\mathcal B}}\chi(z)-\epsilon$, and uses 
$O({\bar d}^{8}\epsilon^{-2}\log t)$ computation steps.\hfill$\square$
\end{lemma}

We now present the algorithm for learning in 
low-dimension. In what follows, we use $\bar d$ instead of $d$ in order to keep the notation consistent with \cref{alg:learnmixture}.

\begin{algorithm}

{\bf{Input:}} IID samples from well-separated mixture of Gaussians in $d=O(\log k)$ dimensions;

{\bf{Output:}} The centers of the components Gaussians;\\

let $count_{max} = C_{4.5} k$\;

{\bf{while}} {$count \leq count_{max}$,}

\begin{itemize}

\item for each point $$\ell \in \left(\frac{\bar\Delta}{1000}  \sqrt{\frac{\bar d}{C_{1.5} \log k}}\right){\mathbb Z}^{\bar d} \bigcap \bigcup_i  B_2(x_i, 2\sqrt{\bar d}),$$ let $Q_\ell$ be the ball of radius $\left(\frac{\bar\Delta}{400}\frac{\bar d}{\sqrt{C_{1.5} \log k}}\right)$, centered at $\ell$\; 

\item use an efficient zeroth order stochastic concave maximization subroutine (see \cref{fact:1}) that produces a point $q_\ell$ in $Q_\ell$ at which $(2\pi)^{-\frac{\bar d}2}(\zeta  \star \mu_e)(z)$ is within $k^{-C_{4.6}}$ of the maximum of $(2\pi)^{-\frac{\bar d}2}(\zeta \star \mu_e)$ restricted to $Q_\ell$\;

\item create a sequence $L = (q_{\ell_1}, \dots, q_{\ell_{k_1}})$ that consists of  all those $q_\ell$, such that $(2\pi)^{- \frac{\bar d}{2}}(\zeta \star \mu_e) (q_{\ell}) > (\frac{w_{min}}{2}) \gamma_{\bar \Delta}(0)$, and  $\|\ell - q_\ell\|_2 < diam(Q_\ell)/4$\;

\item reorder $L$, if necessary, so that $(\zeta \star \mu_e) (q_{\ell_1})\geq\cdots\geq (\zeta \star \mu_e) (q_{\ell_{k_1}})$\;

\begin{itemize}
\item {\bf{if}} {$k_1 <1$,} {\bf{return}} {\textbf{error}}\;

\item {\bf{else}}, form a subsequence $M = (\ell_{m_1}, \dots, \ell_{m_{k_2}})$ of $(\ell_1, \dots, \ell_{k_1})$ using the following iterative procedure\:
\begin{itemize}

\item let $m_1 \leftarrow 1$, $j\leftarrow 1$ and $M := \{\ell_{1}\}$\;

\item[] {\bf{while}} {$j \leq k_1$,}

\item[]\hspace{1cm} {\bf{if}} {$\ell_{j+1}$ is not contained in $B(\ell_{j'}, \sqrt{\bar d}{\sqrt{\Delta \bar\Delta}}/2)$ for any $j' \leq j$,}
\begin{itemize}
\item[]\hspace{1cm} append $\ell_{j+1}$ to $M$\;

\item[]\hspace{1cm} increment $j$\;
\end{itemize}

\end{itemize}

\end{itemize}
\item pass $$\{(q_{\ell_{m_j}}) \}_{j \in [k_2]}$$ to \cref{alg:learnmixture}\;

\item increment $count.$

\end{itemize}
\caption{FindSpikes}\label{alg:findspikes}
\end{algorithm}

\section{Estimating Weights}\label{www}
We find estimates for the weights of the Gaussian components, 
within accuracy $cw_{min}$, when the dimension satisfies $\bar d\leq O(\log k)$. The idea is that the 
weight of a Gaussian component in a well-separated 
mixture is concentrated near the center of the component, 
hence, can be estimated by considering the integral of the 
mixture in an appropriately sized ball around the 
(approximate) center of the component. We formulate this in the following. A proof 
appears in \cref{p00}.

\begin{proposition}[Estimating the weights]\label{weight}
Let $q_\ell$ be an output of \cref{alg:boost}, and let $\hat Q_\ell\subseteq \mathbb R^{\bar d}$ be the ball 
of radius $C_{3.3}\bar\Delta\sqrt{\bar d}$, centered at $\ell$, with volume 
$v$. Let $z_1,\cdots,z_p$ be independent uniformly random points 
from $\hat Q_{\ell}$; if \begin{eqnarray}p\geq \frac{C_{3.5}}{c_{0.1}\delta^2}
k^{\frac {C_{1.5}}2\ln\left(eC_{3.3}^2\right)}
\ln k\,.\end{eqnarray} then 
\begin{eqnarray}\mathbb P\left[\left|\frac 1p\sum_{j=1}^p\Re (f_{z_j})-\int_{\hat Q_\ell}
\gamma_{\bar\Delta}\star\nu~dx\right|>\delta+3k^{-C_{3.5}}v\right]&\leq&20k^{-C_{3.5}}\end{eqnarray} In 
particular, if $y_{j_0}\in Q_\ell$, then \begin{eqnarray}\mathbb P\left[\left|\frac 1p
\sum_{j=1}^p\Re (f_{z_j})-w_{j_0}\right|>cw_{min}\right]&\leq&20k^{-C_{3.5}}\,.\end{eqnarray}\qed
\end{proposition}

\section{Discussions}
\subsection{Why brute-force-search may not work even when $\polyy(2^d)$ many samples are available?} 
As stated in the introduction, our results are 
most interesting when $d=O(\log k)$. Perhaps an intelligent approach in this regime would be to 
get $\polyy(k)$ independent samples $X_1,\cdots,X_{k^C}$ from the Gaussian mixture $\mu$, 
and consider an $\epsilon$-net on the union $\bigcup_j B(X_j,C'\sqrt{d})$, and perform an exhaustive 
brute-force search for the centers, using this net. Criterion to decide 
if a point is close to a center of such a mixture of Gaussians is that, in a small neighbourhood, 
the point is close to a maxima of the mixture density. How we ascertain whether 
a point is close to a local maxima is another story (and we discuss this 
in this paper), but for now, we consider ourselves just powerful enough to be able to 
do this. Note, however, that when $d=\Omega(\log k)$, such an $\epsilon$-net 
has cardinality of potential order 
$\Omega(k^C(C'/\epsilon)^dd^{\frac d2})=
\Omega(k^{C+\log(C'/\epsilon)}k^{\frac 12\log\log k})=\Omega(k^{\log\log k})$. 
Thus, to learn the centers, we 
need to be ready for $\Omega(k^{C+\log(C'/\epsilon)}k^{\frac 12\log\log k})$ 
many searches, which fails to be executed in $\polyy(k)$ time. Moreover, even if one 
ignores $\log\log k$ in the exponent above, there is the so-called {\em mixture-density dilemma} 
that needs to be dealt with. To discuss this, let us consider a unit sphere $\mathbb S^{d-1}
\subseteq \mathbb R^d$, and let $y_1,\cdots,y_k\in \mathbb S^{d-1}$ be such that 
the atomic-mixture $\nu:=k^{-1}(\delta_{y_1}+\cdots+\delta_{y_k})$ is an approximation 
of the uniform surface measure of the sphere; equivalently, the points $y_1,\cdots,y_k$ 
are equidistributed. Let $\gamma$ be the standard $d$-dimensional spherical Gaussian 
density of unit variance in each coordinate. In this case, the origin is a global maxima of the mixture 
$\nu\star \gamma$. This can be easily seen when the centers of the Gaussians are the roots of 
unity $y_r:=e^{\frac{2\pi ir}{k}}$ for $r\in [k]$ with $k=2n$ and $n$ odd. Let $\gamma_t$ be the 2-dimensional Gaussian 
with mean 0 and variance $t^2$, and let $\mu=\nu\star\gamma_t$. 
One has $\mu(0)=(2\pi t^2)^{-1}e^{-\frac 1{2t^2}}$, and for any $r\in [k]$, 
\begin{eqnarray*}\mu\left(e^{\frac{2\pi ir}k}\right)&=& 
\frac {1}{4n\pi t^2}+\frac 1{4n\pi t^2}\sum_{r=1}^{2n-1}
e^{-\frac{\left|1-e^{\frac{\pi ir}n}\right|^2}{2t^2}}\\ &=&
\frac {1+e^{-\frac 2{t^2}}}{4n\pi t^2}+\frac {e^{-\frac 1{t^2}}}{2n\pi t^2}\sum_{r=1}^{n-1}
e^{\frac{\cos\left({\frac{\pi r}n}\right)}{2t^2}}\\ &\leq &
\frac {1+e^{-\frac 2{t^2}}}{4n\pi t^2}+\frac {(n-1)e^{-\frac 1{t^2}}}{4n\pi t^2}
\left(e^{\frac{\cos\left({\frac{\pi}n}\right)}{2t^2}}+
e^{-\frac{\cos\left({\frac{\pi}n}\right)}{2t^2}}\right)\,.\end{eqnarray*} A simple 
observation is that $1+e^{-\frac 2{t^2}}<2e^{-\frac 1{2t^2}}$ when $t\geq 1$; 
hence, \begin{eqnarray*}&\left(e^{\frac{\cos\left({\frac{\pi}n}\right)}{2t^2}}+
e^{-\frac{\cos\left({\frac{\pi}n}\right)}{2t^2}}\right)&<2e^{\frac{\cos\left({\frac{\pi}n}\right)}{2t^2}}\,,\\ \Rightarrow & \mu\left(e^{\frac{2\pi ir}k}\right)& < 
\frac {2e^{-\frac 1{2t^2}}}{4n\pi t^2}+\frac {2(n-1)e^{\frac{\cos\left({\frac{\pi}n}\right)}{2t^2}}
e^{-\frac 1{t^2}}}{4n\pi t^2}\\ &&< (2\pi t^2)^{-1}e^{-\frac 1{2t^2}}\,.\end{eqnarray*} In particular, the center appears to be a global maxima. This example illustrates the short-comings of naive brute-force-search. 

We deal with this, by deconvolving the Gaussian mixture followed by local search. It appears 
from the simple computation that --- as $t\rightarrow 0^+$, the situation should turn around dramatically, 
and this is precisely what is expected, since Gaussian kernel is an efficient approximate identity. Thus, 
suitable application of deconvolution technique results in a mixture that 
approximatey resembles the underlying atomic mixture after smoothning with a thin Gaussian. 
Such a smooth mixture inherits the property of the underlying atomic mixture, which 
prohibits points away from all the centers to manifest as maxima (local or global) of the mixing density. 
We note that the deconvolution itself involves technicalities, since this amounts to multiplying with 
exponentially increasing function in Fourier domain, and in presence of noise in 
computing the characteristic function using samples from the mixture, one needs 
to be sufficiently careful.

\subsection{About assumptions on uniform variance, and almost `uniform' mixing weights}
Our motivation for considering this somewhat stylized setting --- of 
almost `uniform' mixing weight, and Gaussian components having 
uniform variance --- stems from questions in manifold learning. Here, data 
originate from an underlying compact smooth manifold --- having known positive lower-bound 
on the reach and known upper-bound on Hausdorff volume --- isometrically embedded into 
a possibly exponentially large dimensional ambient space, 
and then get noisy by addition of independent Gaussian noise from the ambient 
space. The original distribution supported on the manifold usually 
posses almost uniform Radon-Nikodym derivative with respect to 
the Hausdorff measure on the manifold. Considering the set of centers 
of the individual Gaussian components as compact (nonsmooth) 0-dimensional manifold, 
the above facts account for uniform variance, and almost `uniform' weight, respectively. 
We note that the usual manifold learning results currently available do not carry forward 
to our 0-dimensional settings, mainly because (1) the 0-dimensionality does not 
allow smooth analysis available for positive dimensional manifolds, and (2) the 0-dimensional 
case affords noise with comparatively large variance than in positive dimension.

\subsection{Why is the focussed dimensionality interesting?}\label{subsection:ela1}

One of our practical technical motivation for considering the regime $d\leq \log k$ is 
that essentially the entire analysis on higher dimension reduces to this regime by 
an application of Johnson-Lindenstrauss. This is prevalent even in the more general setting of 
manifold learning, where the noiseless data comes from and underlying 
space --- which is a manifold of dimension $d$ --- and gets corrupted by Gaussian noise: for example, 
in case of Gaussian mixture learning, the underlying space consists of the means of the Gaussians, a zero dimensional 
manifold. We note that this also partially explains the reason for considering 
mixture of spherical Gaussian with identical variance. 
In the general case, a standard (and reasonably well-motivated --- 
by an urge to avoid {\em ``curse of dimensionality"}) trick is 
dimension reduction. If the manifold has standard $d$-dimensional volume $\mu_{I}$ and 
curvature based volume $\mu_{I}'$, and if 
an Euclidean tubular neighbourhood of radius $\tau>0$ of the manifold remains 
reasonably smooth, then Johnson-Lindenstrauss type theorem (see, for example, theorem 1.5 of \cite{MR1500161}), 
asserts that a random projection onto a subspace of dimension of order \begin{displaymath}\epsilon^{-2}
\left(d\log(1/\epsilon)+\log(1/\delta)\right)+\epsilon^{-2}(\log (\mu_{I}/\tau^d+\mu_{I}'))\end{displaymath} 
is an $\epsilon$-approximate isometry with probability at least $1-\delta$. Thus, one 
essentially projects data onto a randomly choosen subspace of dimension specified above. In this paper too, 
we follow this broad strategy, whereby data from high dimensional space are first projected onto 
a sequence of subspaces and the resulting mean coordinates are carefully patched-up.

Many of the popular algorithms for learning mixture of Gaussians, with 
separation assumption --- that appears to be at par or even weaker than the separation assumption in this paper 
--- are based on EM methodology and/or moment-based 
methods. While these algorithms are highly efficient in learning mixture of Gaussians, 
it is not immediately obvious whether they generalize to the more involved case of 
data analysis, where the underlying space is a manifold of positive dimension. On 
the other hand, our approach is generic in some sense, and we expect this can be extended 
to the more general case of manifold learning.

\subsection{Handling mixtures with non-spherical Gaussians}\label{subsection:100}

The main idea in Fourier analytic approach to learning Gaussian mixture 
hinges upon the fact that for mixture of spherical Gaussians with known variance, 
Fourier transform turns the convolution into a product of (sum of) 
sinusoids and a heat kernel, which we know precisely --- so that we can deconvolve. 
In case of non-spherical Gaussians, following obstructions manifests: 
\begin{itemize}
\item in case of unknown variance, we need to first learn the variance, and we don't 
know if it can be done in the realm of Fourier analysis; 
\item in the non-spherical case, the Fourier transform of the density of the mixture becomes 
a non-uniform weighted sum of sinusoids, and a uniform deconvolution is not possible.
\end{itemize}

Fortunately, their are standard techniques in statistics literature that deals with 
covariance estimation for Gaussian mixture model; we refer to \cite{MR2985942} 
and the references cited therein for more details on these aspects.

\section{Conclusion and open problems}\label{conclusion}

We developed a randomized algorithm that learns the centers $y_i$ 
of standard Gaussians in a mixture, to within an $\ell^2$ distance of $\delta < \frac{\Delta\sqrt{d}}{2}$ 
with high probability when the minimum separation between two centers is at least $\sqrt{d}\Delta$, where 
$\Delta$ is larger than an universal constant in $(0, 1)$. 
The number of samples and the computational time is bounded above by 
$\polyy\left(k, d,\log\left(\frac{1}{\delta}\right)\right)$. Such a polynomial bound on the sample and 
computational complexity was previously unknown when $d \geq  \omega(1)$. 
There is a matching lower bound due to 
\cite{MR3734220} on the sample complexity of learning a random mixture 
of Gaussians in a ball of radius $\Theta(\sqrt{d})$ in $d$ dimensions, 
when $d$ is $\Theta( \log k)$. It remains open whether (as raised in \cite{MR3734220}) 
$\polyy(k, d, 1/\delta)$ upper bounds on computational complexity of this task 
can be obtained when the minimum separation between two centers in 
$\Omega(\sqrt{\log k})$ in general, although when $d \leq O(\log k)$, 
this follows from our results. It would also be interesting to extend the 
results of this paper to mixtures of spherical Gaussians whose variances are not necessarily equal, 
or even more, to the case of non-spherical Gaussian mixture.

\bibliography{template}
\appendix

\section{Technical Discussions}
\subsection{Gaussian concentration}
In all the discussions, we will employ the following tail-bound on $\chi^2$-statistic; 
this has appeared as inequality (4.3) of \cite{MR1805785} which contains a proof of the statement.

\begin{lemma}\label{lem:3.0}
Suppose that $u$ is $\chi^2$-statistic with $d$ degrees of freedom; then, for any $x>0$, 
the following inequality holds: \begin{displaymath}\mathbb P\left[u-d\geq 2\sqrt{dx}+2x\right]
\leq e^{-x}\,.\end{displaymath}\qed
\end{lemma}

\begin{remark}\label{rem0}
The reason that we are able to learn the centers (and the weights too, in suitable cases) 
even with Gaussian smoothing at just a constant scale (recall that we are basically scaling 
$\Delta$ by $C_{3.2}$) is hidden in the Gaussian concentration; the probability the scaled Gaussian 
exceeds $CC_{3.2}^{-1}d$ far from its center is at most $e^{-C^2C_{3.2}d}$.\qed
\end{remark}

\begin{remark}\label{rem1}
We note that \cref{lem:3.1} allows us to restrict onto any one particular cluster of centers, 
and we know this cluster of centers has (with high probability) diameter at most $$\rho:=2k\sqrt{3d\ln (C_1dn)}\,.$$ 
Equivalently, we are, from here on, focussing on a $\rho$-bounded mixture 
in dimension $d\leq k$, {\em a la Regev-Vijayaraghavan in \cite{MR3734220}}. For this, it is enough for us 
to estimate the centroid of the mixture by averaging $\polyy(k,d)$ many samples, 
and apply an affine-transform so that origin is placed at this estimated 
centroid. Next, we take the push forward 
of $\mu$ via an orthoprojection of $\mathbb R^d$ on to $S_k$, and work with 
that instead. This allows us to reduce the dimension $d$ to $k$ (if $d$ 
started out being greater than $k$), while retaining the same $\Delta$ to 
within a multiplicative factor of $2$.\qed
\end{remark}

\subsection{Random projection}
We now recall the Johnson-Lindenstrauss lemma for dimension reduction; this will 
be indispensable for dealing with high-dimensional mixture. A proof of the lemma can 
be found in \cite{MR3837109} (see theorem 5.3.1).

\begin{lemma}[Johnson-Lindenstrauss]\label{lem:j-l}
Let ${\mathcal Y}$ be a set of at most $k$ points in $\mathbb R^d$, and $\epsilon>0$. 
Assume that \begin{displaymath}m\geq \left(\frac{\tilde C}{\epsilon^2}\right)
\log k.\end{displaymath} Let $E\subseteq \mathbb R^d$ 
be a random subspace drawn from the uniform distribution on Grassmannian $G_{d,m}$ --- 
the space of $d$-dimensional vector subspaces of $\mathbb R^m$; if $\Pi_E$ 
is the orthogonal projection onto $E$, then with probability at least $1-\exp(-\tilde c\epsilon^2m)$, the 
scaled projection \begin{displaymath}\Pi:=\sqrt{\frac dm}\Pi_E\end{displaymath} is an 
approximate identity on ${\mathcal Y}$: \begin{displaymath}(1-\epsilon)
||y_j-y_l||_2\leq ||\Pi(y_j-y_l)||_2\leq (1+\epsilon)||y_j-y_l||_2\quad
\forall~y_j,y_l\in {\mathcal Y}\,.\end{displaymath}\qed
\end{lemma}

\subsection{Remarks on \cref{randomsample}}
We needed to employ Hoeffding's 
inequality for $\mathbb C$-valued random variables:

\begin{lemma}[Hoeffding]\label{hoeff1}
Let $b>0$. Let $Y_1,\cdots,Y_r$ be independent identically distributed $\mathbb C$-valued 
random variables such that $|Y_j|\leq b^{-1}$ for all $j\in [r]$. Then \begin{eqnarray}\label{hoeff2}\mathbb P\left[\left|
\frac 1r\sum_{j=1}^r(\mathbb E[Y_j]-Y_j)\right|\geq t\right]\leq 2e^{-c_{0.1}rt^2b^2}\end{eqnarray} where 
$c_{0.1}$ is a universal constant.\qed
\end{lemma}

\begin{remark}
The main point to be noted --- \textit{vis-a-vis} \cref{randomsample} --- is that, for {\em each} $x\in\mathbb R^d$, we need polynomially many mixture samples to determine $f_x$; since we need to determine $f_x$ at polynomially many $x\in\mathbb R^d$, this is enough (albeit, slightly wasteful) for the desired sample complexity and time complexity.\qed
\end{remark}

\begin{remark}
We note that, for any $x\in\mathbb R^d$, one has  
\begin{eqnarray*}
\xi_e(x)=
\mbox{vol}({\mathcal B})\mathbb E_z\left[\hat s(z)\exp\left(||z||^2/2\right)\hat\mu_e(z)e^{ix\cdot z}\right] \,.
\end{eqnarray*}
Writing $\xi(z):=\hat s(z)e^{\frac{||z||^2}2}\hat\mu_e(z)e^{ix\cdot z}$, 
one has \begin{eqnarray*}\label{bb1}\mathbb P\left[\left|f_x-\xi_e(x)\right|>k^{-C_{3.5}}\right] &=&
\mathbb P\left[\frac{\mbox{vol}({\mathcal B})}m\left|
\displaystyle\sum_{{l\in[m]}}\xi(z_l)-\mathbb E
\left[\xi(z_l)\right]\right|>k^{-C_{3.5}}\right]\,.\end{eqnarray*} Using 
Hoeffding's inequality from \cref{hoeff2} and the fact that $m=C_{3.6}k^{C_{3.6}}\ln k$, 
we get \begin{eqnarray}\label{bb2}\mathbb P\left[\left|f_x-\xi_e(x)\right|>k^{-C_{3.5}}\right] &\leq &
k^{-C_{3.5}}\,.\end{eqnarray}\qed 
\end{remark}

\section{Proofs}
\subsection{Proof of \cref{lem:3.1}}\label{p3.1}
\begin{proof}
By bounds on the tail of a $\chi^2$-random variable with parameter $d$ (see 
\cref{lem:3.0}), the probability that $||x_l-y(x_l)||\geq \sqrt{3d\ln(C_1dn)}$ 
can be bounded from above as follows. 
\begin{eqnarray}\mathbb P\left[||x-y(x)||\geq\sqrt{3d\ln(C_1dn)}\right]
\leq \begin{pmatrix}3e\ln(C_1dn)\end{pmatrix}^{\frac d2}(C_1dn)^{-\frac{3d}2}
\label{e1}\,.\end{eqnarray} The union bound yields \begin{displaymath}\mathbb P
\left[\exists~x_l~\mbox{s.t}~||x_l-y(x_l)||\geq\sqrt{3d\ln(C_1dn)}\right]
\leq n\begin{pmatrix}3e\ln(C_1dn)\end{pmatrix}^{\frac d2}
(C_1dn)^{-\frac {3d}2} \leq\frac {\eta}{100}\,.\end{displaymath}
\end{proof}

\subsection{Proof of \cref{lem:4.2}}\label{p4.2}
\begin{proof}
This is essentially an application of coupon collector's problem.\\

Recall that $w_{min}\geq ck^{-1}$. Suppose that $m>k$ random 
independent $\mu$-samples, denoted as $x_1,\cdots,x_m$, have 
been picked up. Let $C:=C(k)$ be a positive integer valued function. 
For any $l\in[k_0]$, the probability -- that $x_1,\cdots,x_m$ contain 
no $\mu_l$-sample -- is at most $\left(1-ck^{-1}\right)^m\leq 
e^{-mck^{-1}}$. Thus, the probability -- that $x_1,\cdots,x_m$ 
contain no sample from at least one Gaussian component in 
$\mu$ -- is at most $e^{\ln k-mck^{-1}}$. We ensure this 
probability is at most $\frac{\eta}{300C}$ by having \begin{displaymath}m\geq 
\frac kc\left\lceil\log\left(\frac{300Ck}\eta\right)\right\rceil\,.\end{displaymath} It 
follows that, if at least \begin{eqnarray}n_0:=
\frac {Ck}c\left\lceil\log\left(\frac{300Ck}\eta\right)\right
\rceil\label{101}\end{eqnarray} random independent $\mu$-samples 
were picked up, then with probability at least $1-\frac \eta{300}$ 
all the components were needed to be sampled at least $C$ times.\\

Let $E$ denote the event that $n_0$ random independent $\mu$-samples 
contain at least $C$ points from each Gaussian component. For $l\in[k_0]$, 
let $A_l$ be the event that none of the $n_0$ random samples satisfy 
$||x_j-y_l||<2\sqrt d$. Applying Gaussian isoperimetric inequality 
(see \cref{lem:3.0}) and union bound,
we obtain \begin{eqnarray*}\mathbb P\left[\bigcup_{l\in[k_0]}A_l~
\middle|~E\right]\leq  2^{\log k-Cd}\,.\end{eqnarray*} Thus, letting 
$\label{102}C\geq\log\left(\frac{300k}{\eta}\right)$, it 
follows that \begin{eqnarray*}\mathbb P\left[\bigcap_{l\in[k]}A_l^c\right]&
\geq \mathbb P\left[\bigcap_{l\in[k_0]}A_l^c~\middle|~E\right]\mathbb P
\left[E\right]\\ &\geq\left(1-\frac \eta{300}\right)^2\,,\end{eqnarray*} provided 
\begin{displaymath}n\geq \frac kc\log\left\lceil\frac{300k}{\eta}\right\rceil\left\lceil\log\left
\lceil\frac{300k\log\left\lceil\frac{300k}{\eta}\right\rceil}\eta\right\rceil\right\rceil\,.\end{displaymath}
\end{proof}

\subsection{Proof of \cref{lem:3}}\label{p3}
\begin{proof}
Let $d\hat{\kappa}^2$ be a $\chi^2$ random variable with $d$ degrees of freedom, $i.e.$, the sum of squares of $d$ independent standard Gaussians. By the inequality in \cref{lem:3.0}, we know that
\begin{eqnarray}{\mathbb P}\left[\hat{\kappa} \geq \sqrt{\frac{C_{3.5}\ln k}{d}}+ 1\right]
&={\mathbb P}\left[d\hat{\kappa}^2-d \geq C_{3.5}\ln k+2\sqrt{C_{3.5}d\ln k}\right]\nonumber
\\ & \leq \exp\left(-\frac{C_{3.5}\ln k}2\right)\nonumber\\ &= k^{-\frac{C_{3.5}}2}\,.\label{eq:1}\end{eqnarray} For $z\in \mbox{supp}(\hat s)$, one has $\hat s(z)={\bar\Delta}^{-d}
\gamma_{(1/{\bar \Delta})}(z)$. Therefore, \begin{eqnarray*}\|\hat s - {\bar\Delta}^{-d}
\gamma_{(1/{\bar \Delta})}\|_{L^1}&=\int_{||{\bar\Delta}z||\geq \sqrt{C_{3.5}\ln k}+\sqrt d}
\gamma({\bar\Delta}z)~dz\\ &= {\bar\Delta}^{-d}\mathbb P\left[\hat{\kappa} \geq \sqrt{\frac{C_{3.5}\ln k}{d}}+ 1\right]\\ 
&\leq {\bar\Delta}^{-d}k^{-\frac{C_{3.5}}2}\,.\end{eqnarray*} The 
lemma now follows from the fact that ${\mathcal F}^{-1}$ is a bounded linear operator 
from $L^1$ to $L^\infty$ with operator norm $\left(2\pi\right)^{-\frac d2}$.
\end{proof}

\subsection{Proof of \cref{randomsample}}\label{p0}
\begin{proof}
For any $x\in\mathbb R^d$, one has \begin{eqnarray} &\left|(\gamma_{\bar\Delta}
\star\nu)(x)-(2\pi)^{- d}\sum_{j=1}^{k_0}w_j\int_{\mathcal B}
e^{-\frac{\bar\Delta^2||z||^2}2}e^{i(x-y_j)\cdot z}~dz\right|\nonumber\\ &\leq 
(2\pi)^{- d}\int_{\mathbb R^{\bar d}\setminus {\mathcal B}}
e^{-\frac{\bar\Delta^2||z||^2}2}~dz\nonumber\\  
&\leq k^{-\frac{C_{3.5}}4}\,.\label{hoeff0}\end{eqnarray} Hence, 
it suffices to estimate \begin{displaymath}I_x:=
(2\pi)^{-d}\sum_{j=1}^{k_0}w_j\int_{\mathcal B}e^{-\frac{\bar\Delta^2
||z||^2}2}e^{i(x-y_j)\cdot z}~dz\,.\end{displaymath} Next, one has \begin{eqnarray*}I_x&
= (2\pi)^{-\frac{ d}2}\int_{\mathcal B} e^{ix\cdot z}\left({\hat s}(z)\sum_{j=1}^{k_0}w_je^{-iy_j\cdot z}
\right)~dz\\ &= (2\pi)^{-\frac{ d}2}\mbox{vol}(B)~\mathbb E_z\left[e^{ix\cdot z}\left({\hat s}(z)
\sum_{j=1}^{k_0}w_je^{-iy_j\cdot z}\right)\right]\,,\end{eqnarray*} where 
$z$ is a sample from the  uniform probability distribution on ${\mathcal B}$. For brevity of notation, 
write \begin{eqnarray*}\phi(z)&=
(2\pi)^{-\frac{ d}2}\mbox{vol}({\mathcal B})\left({\hat s}(z)\sum_{j=1}^{k_0}w_je^{-iy_j\cdot z}
\right)\\ &=\mbox{vol}({\mathcal B}){\hat s}(z){\hat \nu}(z)\,,\end{eqnarray*} so that $I_x=
\mathbb E_z[e^{ix\cdot z}\phi(z)]$. By Ramanujan's approximation of $\Gamma$ 
(see \textit{theorem 1} of \cite{MR1850542}) we get  
\begin{eqnarray}\mbox{vol}({\mathcal B})&={\bar\Delta}^{-d}
\left(1+\sqrt{\frac{C_{3.5}\ln k}d}\right)^d\frac{(d\pi)^{\frac d2}}{\Gamma\left(\frac d2+1\right)}\nonumber\\
&\leq (2\pi)^{\frac d2}k^{C_{1.5}\ln\left(\frac {2C_{3.2}}c\right)+
C_{3.5}}\,,\end{eqnarray} so that \begin{eqnarray*}|\phi(z)|&\leq (2\pi)^{-\frac{ d}2}
\mbox{vol}({\mathcal B})~|{\hat s}(z)|\\ &< 
k^{C_{1.5}\ln\left(\frac {2C_{3.2}}c\right)+C_{3.5}}\\ &\leq k^{2C_{3.5}}\,.\end{eqnarray*}
If $z_1,\cdots,z_m$ are independent (Lebesgue) uniformly distributed points in ${\mathcal B}$, then 
by Hoeffding's inequality from \cref{hoeff2}, one has \begin{eqnarray}\label{hoeff6}\mathbb P\left[\left|I_x-
\frac 1m\sum_{l=1}^me^{ix\cdot z_l}\phi(z_l)\right|\geq k^{-2C_{3.5}}\right]
\leq k^{-2C_{3.5}}\,,\end{eqnarray} if we set \begin{eqnarray}\label{m1}m
= C_{0.1}C_{3.5}k^{8C_{3.5}+\frac{2C_{3.5}C_{3.2}^2}{c^2}+2
C_{1.5}\ln\left(\frac{2C_{3.2}}c\right)+C_{1.5}\ln(2\pi)}\ln k.\end{eqnarray} Recalling that 
$||\hat\nu||_{L^\infty},||\hat\mu_e||_{L^\infty}\leq (2\pi)^{-\frac{ d}2}$, and \begin{displaymath}\hat\nu(z)
=\mathbb E_{X\approx\mu}\left[(2\pi)^{-\frac{ d}2}e^{-iX\cdot z+\frac{||z||^2}2}
\right]=\mathbb E\left[e^{\frac{||z||^2}2}
\hat\mu_e(z)\right],\end{displaymath} for any fixed $z\in {\mathcal B}$, applying 
Hoeffding's inequality from \cref{hoeff2} once again, we obtain 
\begin{eqnarray*}\label{hoeff3}\mathbb P\left[\left|\hat\nu(z)-
e^{\frac{||z||^2}2}\hat\mu_e(z)\right|>\frac{k^{-C_{3.5}}}{\mbox{vol}({\mathcal B})}\right]
<\exp\left(-\frac{c_{0.1}n}{(2\pi)^{-\frac d2}k^{2C_{3.5}+\frac{2C_{3.2}^2C_{3.5}}{c^2}}
\mbox{vol}^2({\mathcal B})}\right)\,.\end{eqnarray*} In 
particular, letting \begin{displaymath}C_{3.6}\geq C_{0.1}C_{3.5}+C_{1.5}\left(\ln(2\pi)+2\ln\left(\frac{2C_{3.2}}c\right)\right)+C_{3.5}\left(8+\frac{2C_{3.2}^2}{c^2}\right)\end{displaymath} and \begin{eqnarray}\label{sample1}n&\geq
C_{3.6}k^{C_{3.6}+C_{3.5}}\ln\left(C_{3.6}k^{C_{3.6}+C_{3.5}}\ln k\right)\nonumber\\ m
&=C_{3.6}k^{C_{3.6}}\ln k\,,\end{eqnarray} one has 
\begin{eqnarray}\label{hoeff8}\mathbb P\left[\left|\hat\nu(z)-
e^{\frac{||z||^2}2}\hat\mu_e(z)\right|>\frac{k^{-C_{3.5}}}{\mbox{vol}({\mathcal B})}\right]
<m^{-1}k^{-C_{3.5}}\,.\end{eqnarray} By an application of the union bound, inequalities 
\eqref{hoeff0}, \eqref{hoeff6}, and \eqref{hoeff8} give us this proposition.
\end{proof}

\subsection{Proof of \cref{weight}}\label{p00}
\begin{proof}
Each output $q_l$ of \cref{alg:boost} (as mentioned in \cref{alg:findspikes} above) 
satisfies the inequality $||q_l-l||_2<\frac 14{diam}(Q_\ell)$. Let 
$\hat Q_\ell$ be a ball of radius $C_{3.3}\bar\Delta
\sqrt{\bar d}$, centered at $\ell$. We will choose $C_{3.3}=10^{-3}C_{3.2}
$ appropriately. Suppose that $y_{j_0}$ be the center 
such that $||y_{j_0}-q_\ell||_2<\frac 12 diam(Q_\ell)$; 
such a center may or may not exist, but we are only interested 
when it does exist. Without loss of generality (by translating the mixture
appropriately), we may assume 
that $y_{j_0}=0$.\\

Let $\lambda$ be the Lebesgue measure on $\mathbb R^{\bar d}$. We 
would like to have affirmative answers to each of the 
following three questions: \texttt{a)} if we can well-approximate 
\begin{displaymath}w_{j_0}=\int_{\mathbb R^{\bar d}}\gamma_{\bar\Delta}
\star\nu_{j_0}~d\lambda\end{displaymath} using the restricted 
integral \begin{displaymath} ^1\hat w=\int_{\hat Q_\ell}\gamma_{\bar\Delta}
\star\nu_{j_0}~d\lambda\,;\end{displaymath} \texttt{b)} 
if we can approximate $^1\hat w$ using 
the mixture integral \begin{displaymath} ^2\hat w=\int_{\hat Q_\ell}
\gamma_{\bar\Delta}\star\nu~d\lambda\,;\end{displaymath} and 
\texttt{c)} if we can approximate $^2\hat w$ using Monte-Carlo.\\

The first approximation amounts to bounding the following: \begin{eqnarray*}
\left|w_{j_0}-\int_{\hat Q_\ell}\gamma_{\bar\Delta}\star\nu_{j_0}~d\lambda
\right|&=&\int_{\mathbb R^{\bar d}\setminus \hat Q_\ell}
\gamma_{\bar\Delta}\star\nu_{j_0}~d\lambda\\
&=&w_{j_0}\int_{||x||\geq C_{3.3}\sqrt{\bar d}}\gamma_1
~d\lambda\,.\end{eqnarray*} By Gaussian concentration results, 
we have \begin{eqnarray}\left|w_{j_0}-\int_{\hat Q_\ell}
\gamma_{\bar\Delta}\star\nu_{j_0}~d\lambda
\right|&\leq &w_{j_0}C_{3.3}^{\bar d}e^{-\frac{\bar d(C_{3.3}^2-1)}2}\nonumber\\
&\leq&\frac {cw_{\min}}{10}\end{eqnarray} as soon as \begin{eqnarray}C_{3.3}\geq 
3\ln\left(\frac 1{10c}\right)\,.\end{eqnarray}

For the second approximation, we follow the general strategy 
as in \cref{lem5}. Letting $\omega_{\bar d}$ be the 
volume of unit ball in $\mathbb R^{\bar d}$, and writing $p_m:=|S_m|$ where 
\[S_m:=\begin{pmatrix}r\in [k_0]:~\frac m2\Delta\sqrt{\bar d}
\leq ||y_r||<\frac{m+1}2\Delta\sqrt{\bar d}\end{pmatrix},\] one 
has \begin{eqnarray*}\omega_{\bar d}\left(\frac {\Delta\sqrt{\bar d}}
2\right)^{\bar d}p_m&\leq& \omega_{\bar d}\left(\frac {\Delta\sqrt{\bar d}}
2\right)^{\bar d}\left(\left(m+2\right)^{\bar d}-\left(m-1\right)^{\bar 
d}\right)\\ \Rightarrow\hspace{1cm}p_m&\leq& 
\left(\left(m+2\right)^{\bar d}-\left(m-1\right)^{\bar 
d}\right)\\ &<& \left(m+2\right)^{\bar d}\,,
\end{eqnarray*} which gives \begin{eqnarray}
&&\left|\int_{\hat Q_\ell}\gamma_{\bar \Delta} \star \nu
~d\lambda-\int_{\hat Q_\ell}\gamma_{\bar \Delta} \star \nu_{j_0}
~d\lambda\right|\nonumber\\ &=&(2\pi\bar\Delta^2)^{-\frac {\bar d}2}\sum_{m=2
}^\infty \sum_{r\in[k_0]\setminus\{j_0\},  
\frac m2\leq ||\frac{y_r}{\Delta\sqrt{\bar d}}||< \frac{m+1}2}
w_r\int_{\hat Q_\ell}e^{-\frac{C_{3.2}^2||x-y_r||^2}
{2\Delta^2}}~dx\nonumber\\ 
&< &(ck)^{-1}\sum_{m=1}^\infty e^{-\frac{\bar dC_{3.2}^2(248m)^2}2}\nonumber\\
&< &(ck)^{-1}\sum_{m=1}^\infty e^{-30000\bar dC_{3.2}^2m^2}\nonumber\\
&<&\frac{cw_{min}}{10}\frac {10}{c^2e^{30000\bar dC_{3.2}^2}}\,.\end{eqnarray}

We proceed with the third approximation as follows. 
Let $v:=\mbox{vol}(\hat Q_\ell)$; 
by Ramanujan's approximation of $\Gamma$ (see 
\textit{theorem 1} of \cite{MR1850542} 
) we get \begin{displaymath}v\leq \begin{pmatrix}{\frac{2\pi eC_{3.3}^2}{\bar d}}
\end{pmatrix}^{\frac {\bar d}2}{\bar d}^{\frac {\bar d}2}
\bar\Delta^{\bar d}\leq \bar\Delta^{\bar d}k^{\frac 12C_{1.5}\log(2\pi eC_{3.3}^2)}\,.\end{displaymath} Consider 
the identity \begin{eqnarray*}\int_{\hat Q_\ell}\gamma_{\bar\Delta}\star\nu~dx&=
&{\mathbb E}_{U_{\hat Q_\ell}}\left[v\gamma_{\bar\Delta}\star
\nu\right]\,,\end{eqnarray*} where $U_{\hat Q_\ell}$ is the uniform 
distribution on $\hat Q_\ell$. Let $z_1,\cdots,z_p\in \hat Q_\ell$ be 
uniformly random points; since $||\gamma\star \nu||_{L^\infty}\leq 
(2\pi\bar\Delta^2)^{-\frac {\bar d}2}$, an application of Hoeffding's 
\cref{hoeff1} implies \begin{eqnarray*}\mathbb P\left[
\left|\frac 1p\sum_{j=1}^p\gamma_{\bar\Delta}\star\nu(z_j)
-{\mathbb E}_{U_{\hat Q_\ell}}\left[\gamma_{\bar\Delta}\ast\nu\right]
\right|>\frac{\delta}v\right]&\leq&e^{-\frac{c_{0.1}(2\pi\bar\Delta^2)^{\frac d2}p
\delta^2}{v^2}}\,.\end{eqnarray*} Thus, for \begin{displaymath}p\geq 
\frac{C_{3.5}}{c_{0.1}\delta^2}k^{\frac {C_{1.5}}2\ln\left(eC_{3.3}^2\right)}
\ln k,\end{displaymath} an application of \cref{randomsample} and a union bound 
implies \begin{eqnarray*}&&\mathbb P
\left[\left|\frac 1p\sum_{j=1}^p\Re (f_{z_j})-{\mathbb E}_{U_{\hat Q_\ell}}
\left[\gamma_{\bar\Delta}\star\nu_j\right]\right|>\frac \delta v+
3 k^{-C_{3.7}}\right]\\ &\leq& \mathbb P
\left[\exists~j\in [p]~\mbox{s.t}~\left|\Re (f_{z_j})-\gamma_{\bar\Delta}\star\nu(z_j)\right|>
3k^{-C_{3.7}}\right]\\ &&\hspace{1cm} +~\mathbb P
\left[\left|\frac 1p\sum_{j=1}^p\gamma_{\bar\Delta}\star\nu(z_j)-{\mathbb E}_{U_{\hat Q_\ell}}
\left[\gamma_{\bar\Delta}\star\nu_j\right]\right|>\frac\delta v\right]\\
&\leq & 20k^{-C_{3.5}}\,.\end{eqnarray*} 

We write $\tilde\nu_{j_0}=\nu-\nu_{j_0}$, and \begin{displaymath}\hat w_{\ell}:=\frac vp\sum_{j=1}^p
\Re (f_{z_j})\,,\end{displaymath} where $z_1,\cdots,z_p$ are independent 
uniformly random points in $\hat Q_\ell$. Then 
\begin{eqnarray*}\left|w_{j_0}-\hat w_\ell\right| &
\leq \int_{\mathbb R^{\bar d}\setminus \hat Q_\ell}
\gamma_{\bar\Delta}\star\nu_{j_0}~d\lambda+\left|\int_{\hat Q_\ell}
\gamma_{\bar\Delta}\star\nu~d\lambda-\hat w_\ell\right|
+\left|\int_{\hat Q_\ell}\gamma_{\bar\Delta}\star \tilde\nu_{j_0}~d\lambda\right|\\ &\leq
\frac{cw_{min}}5+\left(\delta+3vk^{-C_{3.7}}\right)\,.\end{eqnarray*}
\end{proof}

\section{Analysis of \cref{alg:findspikes}}\label{an:random}
The following lemma shows that the number of cubes in \cref{alg:findspikes} that need to be considered is polynomial in $k$.
\begin{lemma}\label{lem:4.2.8}
For any $x \in \mathbb R^d$, \begin{displaymath}\left|\left(\frac{\bar\Delta}{1000}  \sqrt{\frac{d}{C_{1.5} \log k} 
}\right)\mathbb Z^{d} \bigcap  B_2(x, 2\sqrt{d})\right| \leq k^{C_7}\,.\end{displaymath}
\end{lemma}
\begin{proof}
Set $r$ to $\sqrt{{d}}$. We observe that the number of lattice cubes of side length $c r/\sqrt{\log k}$ centered at a point $x$ belonging to $\left(c r/\sqrt{\log k}\right)\mathbb Z^{d}$ that intersect a ball $B$ of radius $r$ in $d$ dimensions is less or equal to to the number of lattice points inside a concentric ball $B'$ of dimension $d$ of radius $r + \frac{c \sqrt{d} r}{\sqrt{\log k}}$. Every lattice cube of side length $c r/\sqrt{\log k}$ centered at a lattice point belonging to $\left(c r/\sqrt{\log k}\right)\mathbb Z^{d} \cap B'$ is contained in the ball $B''$ with center $x$ and radius $r + \frac{2c \sqrt{d} r}{\sqrt{\log k}}$. By volumetric considerations, 
$|\left(c r/\sqrt{\log k}\right)\mathbb Z^{d}\cap B'|$ is therefore bounded above by $\frac{vol(B'')}{\left(c r/\sqrt{\log k}\right)^{d}}$.
We write $v_{d}:=\mbox{vol}\begin{pmatrix}B_2(0,\sqrt {\bar d})\end{pmatrix}$. By Ramanujan's approximation of $\Gamma$ (see \textit{theorem 1} of \cite{MR1850542}) we get \begin{displaymath}v_{d}\leq {\frac 1{\sqrt {\pi}}}\begin{pmatrix}{\frac{2\pi e}{d}}
\end{pmatrix}^{\frac {d}2}d^{\frac {{d}}2}\,.\end{displaymath} This tells us that 
\begin{eqnarray*}\frac{vol(B'')}{\left(c r/\sqrt{\log k}\right)^{d}} & \leq& \left(\frac{C\sqrt{\log k}}{\sqrt{d}}\right)^{d}\\
& \leq& \left(\frac{C\sqrt{\log k}}{\sqrt{d}}\right)^{\left(\frac{\sqrt{d}}{C\sqrt{\log k}}\right)\left(C \sqrt{d \log k}\right)}\\
& \leq& \left(e^{1/e}\right)^{\left(C \sqrt{d \log k}\right)}\\
& \leq& k^{C_{7}}\,.\end{eqnarray*}
We have used the fact that for $\alpha > 0$, $\alpha^{\alpha^{-1}}$ is maximized when $\alpha = e$, a fact easily verified using calculus.
\end{proof}

We will need some results on the structure of $\gamma_{\bar \Delta}\star\nu$. 

\begin{lemma}\label{lem5}
If $q_\ell \in B\left(y_j, diam(Q_\ell)\right)$ for some $j\in [k_0]$, then the restriction of $\gamma_{\bar \Delta} 
\star \nu$ to $Q_\ell$ is approximately log-concave in the following sense. For $\nu_j:=w_j\delta_{y_j}$, and $x \in Q_\ell$, we have \begin{displaymath}
0\leq \log (\gamma_{\bar \Delta} \star \nu)
(x)-\log(\gamma_{\bar \Delta} \star \nu_j)(x)
\leq~ {\frac {C_{103}}{c^2{\bar d}^{C_{100}}}}\,.\end{displaymath}
\end{lemma}

\begin{proof}
Fix $x\in Q_\ell$, and write $a_r:=x-y_r$ for $r\in [k_0]$. One 
has \[||a_j||\leq {\frac{1.1\Delta\sqrt{\bar d}}{100C_{3.2}}}
\sqrt{\frac{\bar d}{C_{1.5}\log k}}+{\frac{\Delta\sqrt{\bar d}}{200C_{3.2}}}
\sqrt{\frac{\bar d}{C_{1.5}\log k}}\leq {\frac{3.2\Delta\sqrt{\bar d}}
{200C_{3.2}}}\,,\] and \[||a_r-a_s||\geq \Delta\sqrt{\bar d}\] if $r\neq s$. 
Rewriting $\Delta=C_{3.2}\bar\Delta$, one has \begin{eqnarray*}(2\pi{\bar
\Delta}^2)^{\frac {\bar d}2}\left|(\gamma_{\bar \Delta} \star \nu)(x)-(\gamma
_{\bar \Delta} \star \nu_j)(x)\right|&=&\sum_{r\in[k_0]\setminus\{j\}}w_r
e^{-\frac{||a_r||^2}{2{\bar\Delta}^2}}\\ &=&\sum_{m=1}^\infty \displaystyle
\sum_{{r\in[k_0]\setminus\{j\},~
\frac m2\leq ||\frac{a_r}{\Delta\sqrt{\bar d}}||< \frac{m+1}2}}w_re^{-\frac{C_{3.2}^2
||a_r||^2}{2\Delta^2}}\,.\end{eqnarray*} We write $p_m:=|S_m|$ where 
\begin{displaymath}S_m:=\begin{pmatrix}r\in [k_0]:~\frac m2\Delta\sqrt{\bar d}\leq ||a_r||<\frac{m+1}2\Delta
\sqrt{\bar d}\end{pmatrix}\,.\end{displaymath} Since $||a_r-a_s||\geq \Delta{\bar d}$, we can put disjoint 
balls of radius $0.5\Delta\sqrt{\bar d}$ around each $a_r$. Thus, letting 
$\omega_{\bar d}$ be the volume of unit ball in $\mathbb R^{\bar d}$, one 
has \begin{eqnarray*}\omega_{\bar d}\left(\frac {\Delta\sqrt{\bar d}}
2\right)^{\bar d}p_m&\leq& \omega_{\bar d}\left(\frac {\Delta\sqrt{\bar d}}
2\right)^{\bar d}\left(\left(m+2\right)^{\bar d}-\left(m-1\right)^{\bar 
d}\right)\\ \Rightarrow\hspace{1cm}p_m&\leq& \left(\left(m+2\right)^{\bar d}-\left(m-1\right)^{\bar 
d}\right)\\ &<& \left(m+2\right)^d\,,\end{eqnarray*} which gives \begin{eqnarray*}(2\pi{\bar
\Delta}^2)^{\frac {\bar d}2}\left|(\gamma_{\bar \Delta} \star \nu)
(x)-(\gamma_{\bar \Delta} \star \nu_j)(x)\right| &=&\sum_{m=1
}^\infty \sum_{{r\in[k_0]\setminus\{j\},~
\frac m2\leq ||\frac{a_r}{\Delta\sqrt{\bar d}}||< \frac{m+1}2}}w_re^{-\frac{
C_{3.2}^2||a_r||^2}{2\Delta^2}}\\ &\leq &(ck)^{-1}\sum_{m=1
}^\infty e^{-\frac{d(C_{3.2}^2m^2-8\ln \left(m+2\right))}{8}}\\ &\leq &(ck)^{-1}\sum_{m=1
}^\infty e^{-\frac{d(C_{3.2}^2m^2)}{16}}\,.\end{eqnarray*} We 
use $e^{-x}<\frac{C_{101}^{C_{100}}}{x^{C_{100}}}$ and $\sum m^{-2C_{100}}< 10$
to obtain \begin{eqnarray*}(2\pi{\bar\Delta}^2)^{\frac {\bar d}2}\left|(\gamma_{\bar \Delta} \star \nu)
(x)-(\gamma_{\bar \Delta} \star \nu_j)(x)\right| 
&\leq &(ck)^{-1}\sum_{m=1}^\infty e^{-\frac{{\bar d}C_{3.2}^2m^2}{16}}\\ &\leq 
&\frac {(16C_{101})^{C_{100}}}{ck{\bar d}^{C_{100}}C_{3.2}^{2C_{100}}} \sum_{m=1}^\infty \frac 1{m^{2C_{100}}}\\
&\leq &\frac {C_{102}}{ck{\bar d}^{C_{100}}}\,.\end{eqnarray*} Note 
that $\gamma_{\bar\Delta}\star\nu_j$ 
is log-concave. Moreover, for any $x\in Q_l$, one has 
\begin{eqnarray*}(\gamma_{\bar\Delta}\star\nu_j)(x)&=&
(2\pi\bar\Delta^2)^{-\frac{\bar d}2}w_je^{-\frac{||x-y_j||^2}{2{\bar\Delta}^2}}\\
&\geq &{\frac c{k(2\pi{\bar\Delta}^2)^{\frac{\bar d}2}}}e^{-\frac{1.21{\bar d}}{20000C_{1.5}\ln k}}\\
&\geq &{\frac c{k(2\pi{\bar\Delta}^2)^{\frac{\bar d}2}}}e^{-\frac{1.21}
{20000}}\,,\end{eqnarray*} so that \begin{eqnarray*}\left|\frac{(\gamma_{\bar \Delta} \star \nu)
(x)}{(\gamma_{\bar \Delta} \star \nu_j)(x)}-1\right| 
&\leq & \frac 1{(2\pi{\bar \Delta}^2)^{\frac {\bar d}2}}\frac {C_{102}}{ck{\bar d}^{C_{100}}}
{\frac {k(2\pi{\bar\Delta}^2)^{\frac{\bar d}2}}c}e^{\frac{1.21}{20000}}\\
&\leq & {\frac {C_{102}e^{\frac{1.21}{20000}}}{c^2{\bar d}^{C_{100}}}}\,.
\end{eqnarray*} This gives \begin{eqnarray}\label{logcon}0&\leq& \log (\gamma_{\bar \Delta} \star \nu)
(x)-\log(\gamma_{\bar \Delta} \star \nu_j)(x)\nonumber\\
& = &\log\left(1+{\frac {C_{102}e^{\frac{1.21}{20000}}}{c^2{\bar d}^{C_{100}}}}\right)\nonumber\\
&\leq & {\frac {C_{102}e^{\frac{1.21}{20000}}}{c^2{\bar d}^{C_{100}}}}\,.\end{eqnarray}
\end{proof}

\begin{remark} 
Note that, by \cref{rem1} and convexity of $Q_\ell$, if \[
q_\ell \in B\left(y_j, \left(\frac{\bar\Delta \bar d}{200}
\sqrt{\frac{1}{C_{1.5} \log k} }\right)\right)\] then (with high probability) 
$\log (\gamma_{\bar\Delta}\star\nu_j)$ is $t$-Lipschitz 
for \[t\leq \frac{(ck\sqrt{\ln(C_1dn)}+C_{3.2})\sqrt{C_{1.5}\log k}}{c^3}\,.\]\qed
\end{remark}

The following lemma shows that every true spike of $\gamma_{\bar \Delta}
\star \nu$ corresponding to a center gets detected by \cref{alg:findspikes}.
\begin{lemma}\label{lem:5}
If $q_\ell \in B\left(y_j, \left(\frac{\bar\Delta \bar d}{200}
\sqrt{\frac{1}{C_{1.5} \log k} }\right)\right)$ for some $j \in [k_0]$, then  \begin{eqnarray*} (\gamma_{\bar \Delta}
\star \nu)(q_\ell) - k^{- C_{3.5}} \geq  (\frac{w_{min}}{2})\gamma_{\bar \Delta}(0)\,.\end{eqnarray*}\qed
\end{lemma}

\begin{proof}
By \cref{lem5} the restriction to $Q_\ell$ of $\gamma_{\bar \Delta}
\star \nu$ is approximately log-concave. The output $q_\ell$ of stochastic optimization 
algorithm ${\mathfrak A}_0$ satisfies \begin{eqnarray*}\left|\log m_\ell-\log (\gamma_{\bar \Delta} \star \nu)
(q_\ell)\right|\leq {\frac 1{C_{103.1}{\bar d}^{C_{100}}}}\,,\end{eqnarray*} where $m_\ell:=\max\{(\gamma_{\bar \Delta} \star \nu)
(x): x\in Q_\ell\}$. Equivalently, \begin{eqnarray*}m_\ell-(\gamma_{\bar \Delta} \star \nu)(q_\ell)&\leq&
(\gamma_{\bar \Delta} \star \nu)(q_\ell) \left(1-e^{-\frac 1{C_{103.1}{\bar d}^{C_{100}}}}\right)\\ &\leq&
\frac{(\gamma_{\bar \Delta} \star \nu)(q_\ell)}{C_{103.1}{\bar d}^{C_{100}}}\\ \Rightarrow
\hspace{1cm}(\gamma_{\bar \Delta} \star \nu)(q_\ell)&\geq&
m_\ell-\frac{(\gamma_{\bar \Delta} \star \nu)(q_\ell)}{C_{103.1}{\bar d}^{C_{100}}}\end{eqnarray*} 
which proves the lemma.
\end{proof}

The next lemma shows that there are no false spikes in $\gamma_{\bar \Delta} \star \nu$. 

\begin{lemma}\label{lem:8}
If $(\gamma_{\bar \Delta} \star \nu)(q_\ell) + k^{- C_{3.5}} \geq (\frac{w_{min}}{2})
\gamma_{\bar \Delta}(0)$, then there exists some $j\in [k_0]$ such that $q_\ell \in 
B\left(y_j, \frac{\sqrt{\bar d} (\Delta\bar\Delta)^{\frac 12}}{5}\right)$.\qed
\end{lemma}
\begin{proof}
The arguments are similar to that in \cref{lem5} above. 
Suppose that \begin{displaymath}q_\ell \notin \bigcup_{j \in [k_0]} B\left(y_j, \frac{\sqrt{\bar d}
(\Delta\bar\Delta)^{\frac 12}}{5}\right)\,,\end{displaymath} 
and write $a_j:=y_j-q_\ell$. For $m\geq 1$, let $p_m:=|S_m|$ where 
\begin{displaymath}S_m:=\begin{pmatrix}r\in [k_0]:~\left(\frac m2+\frac 1{5\sqrt{C_{3.2}}}\right)\Delta\sqrt
{\bar d}\leq ||a_r||<\left(\frac {m+1}2+\frac 1{5\sqrt{C_{3.2}}}\right)\Delta
\sqrt{\bar d}\end{pmatrix}\,.\end{displaymath} One has \begin{eqnarray*}(\gamma_{\bar \Delta} \star \nu)(q_\ell)
&=&\gamma_{\bar\Delta}(0)\sum_{j\in[k_0]}w_j
e^{-\frac{||a_j||^2}{2{\bar\Delta}^2}}\\ &=&\gamma_{\bar\Delta}(0)\left(\displaystyle
\sum_{{j\in[k_0],~
0\leq ||\frac{a_j}{\Delta\sqrt{\bar d}}||-\frac 1{5\sqrt{C_{3.2}}}< \frac 12}}w_re^{-\frac{
||a_j||^2}{2\bar\Delta^2}}
+\sum_{m=1}^\infty \displaystyle
\sum_{j\in S_m}w_re^{-\frac{
||a_j||^2}{2\bar\Delta^2}}\right)\,.\end{eqnarray*} Since $||a_r-a_s||\geq \Delta{\bar d}$, we can put disjoint 
balls of radius $0.5\Delta\sqrt{\bar d}$ around each $a_r$. Thus, \begin{eqnarray*}
p_m&\leq& 2^{\bar d}\left(\left(\frac{m+2}2+\frac 1{5\sqrt{C_{3.2}}}\right)^{\bar d}
-\left(\frac{m-1}2+\frac 1{5\sqrt{C_{3.2}}}\right)^{\bar d}\right)\\ &<& 
\left(m+2\right)^{\bar d}\,,\end{eqnarray*} which gives \begin{eqnarray*}(\gamma_{\bar \Delta} \star \nu)(q_\ell)
&=&\gamma_{\bar\Delta}(0)\left(\displaystyle\sum_{{j\in[k_0],~
0\leq ||\frac{a_j}{\Delta\sqrt{\bar d}}||-\frac 1{5\sqrt{C_{3.2}}}< 
\frac 12}}w_re^{-\frac{
||a_j||^2}{2\bar\Delta^2}}
+\sum_{m=1}^\infty \displaystyle\sum_{j\in S_m}w_re^{-\frac{
||a_j||^2}{2\bar\Delta^2}}\right)\\ &\leq &\frac{C\gamma_{\bar\Delta}(0)}k\left(e^{-\frac{C_{3.2}
\bar d}{50}}+\sum_{m=1}^\infty e^{-\frac{\bar dC_{3.2}m^2}{16}}\right)\\
&<&\frac{C\gamma_{\bar\Delta}(0)}k\left(e^{-\frac{C_{3.2}
\bar d}{50}}+\frac {2000}{{\bar d}^2C_{3.2}^2}\right)\,.\end{eqnarray*} Thus, 
for $C_{3.2}$ sufficiently large, the inequality \begin{displaymath}(\gamma_{\bar \Delta}
\star \nu)(q_\ell)+ k^{- C_{3.5}}<\left(\frac{w_{min}}{2}\right)
\gamma_{\bar \Delta}(0)\end{displaymath} holds, proving the lemma.
\end{proof}

\begin{lemma}\label{lem:6.1}
If there exists some $j\in [k_0]$ such that $$q_\ell \in 
B\left(y_j, \frac{\sqrt{\bar d} (\Delta\bar\Delta)^{\frac 12}}{5}\right),$$ and $|q_\ell - \ell| < diam(Q_\ell)/4,$ then with high probability, there exists some $j\in [k_0]$ such that $q_\ell \in 
B\left(y_j, c {\bar d}^{-\frac{5}{2}}\right)$.\qed
\end{lemma}

\begin{proof}
If $x \in  B\left(y_j, \frac{\sqrt{\bar d} (\Delta\bar\Delta)^{\frac 12}}{5}\right)$, by \cref{lem5}, 
\begin{displaymath}
0\leq \ln (\gamma_{\bar \Delta} \star \nu)
(x)-\ln(\gamma_{\bar \Delta} \star \nu_j)(x)
\leq {\frac {1}{K{\bar d}^{5}}},\end{displaymath} where $K$ is an absolute constant that can be made arbitrarily large.
We note that $\ln(\gamma_{\bar \Delta} \star \nu_j)(x) = a - \left(\frac{1}{2{\bar \Delta}^2}\right)\|x - y_j\|^2,$ for some constant $a$. This implies that if \begin{displaymath}|\ln (\gamma_{\bar \Delta} \star \nu)(q_\ell) - \sup_x \ln (\gamma_{\bar \Delta} \star \nu)(x)| < \frac {2}{K{\bar d}^{5}},\end{displaymath} then $q_\ell \in B\left(y_j, c {\bar \Delta} {\bar d}^{-\frac{5}{2}}\right)$. However, $| \ln (\gamma_{\bar \Delta} \star \nu)(q_\ell) - \sup_x \ln (\gamma_{\bar \Delta} \star \nu)(x)|$ is indeed less than $\frac {2}{K{\bar d}^{5}}$ by \cref{randomsample} and \cref{fact:1}. Noting that $\bar \Delta < c$, this completes the proof of this lemma.
\end{proof}

The following proposition shows that every spike extracted out of $\gamma_{\bar\Delta}\star\nu$ by \cref{alg:findspikes}, is within $\epsilon/k$ of some  $y_i$.

\begin{proposition} With probability at least $1 - \exp(-k/c)$, the following is true:
The Hausdorff distance between $\{y_1, \dots, y_{k_0}\}$ and $\{\ell_{m_1}, \dots, \ell_{m_{k_2}}\}$ (which corresponds to the output of \cref{alg:boost} as mentioned in \cref{alg:findspikes} above) is less than $\epsilon/k$.\qed
\end{proposition}
\begin{proof}
Consider the sequence $L$ in \cref{alg:findspikes}. Note that, we know from the statements in \cref{lem:5}, \cref{lem:8} and \cref{lem:6.1} that 
\begin{itemize} 
\item  If $q_\ell \in \bigcup_{j \in [k_0]} B\left(y_j, \left(\frac{\bar\Delta \bar d}{200}  \sqrt{\frac{1}{C_{1.5} \log k} }\right)\right)$, then \begin{displaymath}(\nu \star \gamma_{\bar \Delta})(q_\ell) - k^{- C_{3.5}} \geq  \left(\frac{w_{min}}{2}\right)\gamma_{\bar \Delta}(0)\,.\end{displaymath}
\item 
If $(\gamma_{\bar \Delta} \star \nu)(q_\ell) + k^{- C_{3.5}} \geq (\frac{w_{min}}{2})
\gamma_{\bar \Delta}(0)$, then there exists some $j\in [k_0]$ such that $$q_\ell \in 
B\left(y_j, \frac{\sqrt{\bar d} (\Delta\bar\Delta)^{\frac 12}}{5}\right)\,.$$
\item If there exists some $j\in [k_0]$ such that $$q_\ell \in 
B\left(y_j, \frac{\sqrt{\bar d} (\Delta\bar\Delta)^{\frac 12}}{5}\right)\,,$$ and (as is true from 
\cref{alg:findspikes}), $|q_\ell - \ell| < diam(Q_\ell)/4,$ then with probability at least $1 - \exp(\frac{-k}{c})$, there exists some $j\in [k_0]$ such that $q_\ell \in 
B\left(y_j, c {\bar d}^{-\frac{5}{2}}\right)$.\\
\end{itemize}
We have shown that $L$ is, with high probability, contained in $$\bigcup_{j \in [k_0]} B(y_j,  c{\bar d}^{-5/2})\,.$$ 
Moreover, we observe that for each $j\in [k_0]$, there must exist a $q_\ell$ such that $q_\ell \in B\left(y_j,  c{\bar d}^{-\frac{5}{2}}\right)$, by the exhaustive choice of starting points.
This proposition now follows from Theorem 4.1 of \cite{MR3734220}.
\end{proof}

\section*{Acknowledgement} We acknowledge the support of DAE project 12-R\&D-TFR-5.01-0500. 
We thank International Centre for Theoretical Sciences (ICTS) for their hospitality during a visit for the program - Statistical Physics of Machine Learning (ICTS/SPMML2020/01).

\end{document}